\documentclass{article}
\pdfoutput=1     
\usepackage[T1]{fontenc}
\usepackage[latin1]{inputenc}
\usepackage{times}
\usepackage[english]{babel}
\usepackage{amsfonts, amsxtra, amstext, amssymb,amsthm}
\usepackage[numbers]{natbib}
\usepackage{url}
\usepackage{graphicx}
\usepackage[caption=false]{subfig}
\usepackage{booktabs}
\usepackage{multirow}
\usepackage{color}
\usepackage{pgfplots}
\usetikzlibrary{pgfplots.units}
\usetikzlibrary{automata,arrows}
\usepackage{siunitx}
\pgfplotsset{compat=newest}
\pgfplotsset{height=5.5cm}
\pgfplotsset{width=7cm}
\usepackage{colortbl}
\usepackage[all]{xy}
\usepackage{rotfloat}
\usepackage{listings}
\usepackage{arxiv}

\newtheorem{statement}{Statement}
\newtheorem{lemma}{Lemma}
\newcommand{\compactpar}[1]{
\smallskip {\noindent \textbf{#1}}\ \\ }

\newcommand{\term}{\ensuremath{\mathbf{V}_t}}
\newcommand{\nterm}{\ensuremath{\mathbf{V}_n}}
\newcolumntype{d}[1]{D{.}{.}{#1}}

\usepackage{color} 
\lstdefinelanguage{asn.1}
{keywords=%
{DEFINITIONS,ENCODED,CONTAINING,SEQUENCE,OCTET,INTEGER,BEGIN,END,OPTIONAL,
STRING,BIT,OCTET,UNIVERSAL,DEFAULT,SET,BOOLEAN,EXPLICIT,IMPLICIT,OBJECT,
IDENTIFIER,ANY,DEFINED,BY,CHOICE,UTCTime,GeneralizedTime,SIZE,(1..MAX),OF,
IA5String,VisibleString,BMPString,PrintableString,UTF8String,TeletexString},
sensitive=true,
comment= [l]{--},
}[keywords]

\lstdefinelanguage{grammar}
{morecomment = [s]{\{}{\}},
morestring = [b]',
literate=%
{->}{ $\rightarrow$ }3%
{eps}{$\varepsilon$}1
{|[}{{\texttt{|$\!\!$[}}}1%
{]|}{{\texttt{]$\!\!$|}}}1%
{->}{$\rightarrow$ }2%
{=>}{$\Rightarrow$ }2,
morekeywords = {[2]Certificate}}

\title{Systematic Parsing of X.509:\\ 
       Eradicating Security Issues with a Parse Tree
      }
\runningtitle{Systematic Parsing of X.509}
\disclaimer{We remark that this article is a preprint. The final version of 
	this work is published in the Journal of Computer Security, Volume 26, 
	Issue 6, 30th Oct. 2018. DOI: \url{http://dx.doi.org/10.3233/JCS-171110}}

\author{Alessandro Barenghi\Affiliation{1} \\
	alessandro.barenghi@polimi.it
	\And
	Nicholas Mainardi\Affiliation{1} \\
	nicholas.mainardi@polimi.it
	\And
	Gerardo Pelosi\Affiliation{1} \\
	gerardo.pelosi@polimi.it		
	\StartInstitutes
	\Institute{1}{Department of Electronics, Information and Bioengineering 
		\\ Politecnico di Milano, IT}
	\EndInstitutes
}
\begin{document}
\maketitle
\begin{abstract}
X.$509$ certificate parsing and validation is a critical task which has 
shown consistent lack of effectiveness, with practical attacks being reported 
with a steady rate during the last $10$ years.
In this work we analyze the X.$509$ standard and provide a grammar description
of it amenable to the automated generation of a parser with strong 
termination guarantees, providing unambiguous input parsing.
We report the results of analyzing a $11$M X.$509$ certificate dump of the HTTPS 
servers running on the entire IPv$4$ space, showing that $21.5$\% of the 
certificates in use are syntactically invalid.
We compare the results of our parsing against $7$ widely used TLS libraries
showing that $631$k to $1,156$k syntactically incorrect certificates are 
deemed valid by them ($5.7$\%--$10.5$\%), including instances with security 
critical mis-parsings.
We prove the criticality of such mis-parsing exploiting one of the syntactic 
flaws found in existing certificates to perform an impersonation attack.
\end{abstract}
%
%
%
%
%
%
\section{Introduction}\label{sec:intro}
Digital certificates are the mainstay of public key authentication in secure
network communications since the introduction of the Secure Sockets Layer (SSL) 
and Transport Security Layer (TLS) protocols.
The requirement for interoperability called for the adoption of a standardized
format specifying both the information which should be contained in them and
their encoding.

Such a standardization action took place within the X.$500$ standard series
by the International Telecommunication Union (ITU), and resulted in the
X.$509$ standard for digital certificates. 
The standard has been further described in a series of Request for Comments 
(RFCs) memoranda~\cite{RFC3279,RFC4055,RFC4491,RFC5280,RFC5758,RFC5480}
by the Internet Engineering Task Force (IETF) and has grown to a considerable 
complexity since its first version designed in $1988$. 
As a consequence, performing sound validation of the contents of an X.$509$ 
certificate has become a non trivial and security critical task. 
Indeed, X.$509$ certificate validation has shown consistent lack of 
effectiveness in implementations due to the large number of constraints 
to be taken into account and the complexity of the involved data structure.
Practical attacks against the X.$509$ certificate validation have
been pointed out for the last $10$ years, leading to effective impersonations 
against TLS/SSL enabled software.

Some among the most renown security issues involve certificates which are 
deemed valid to be binding a public key to the identity of a Certification
Authority (CA), while such an information is contradicted either by the values 
contained in the certificate~\cite{marlinspike02} or by misinterpretations in 
the subject name contained in it~\cite{marlinspike09}, both leading to 
effective impersonation of an arbitrary identity.
More recently, in~\cite{DBLP:conf/fc/KaminskyPS10} it was shown that 
inconsistent validations were performed by different TLS libraries, 
due to integer overflows in the recognition of some X.$509$ certificate fields,
providing ground for attacks.
Broken certificates are common even among the Alexa top $1$M visited 
sites~\cite{DBLP:conf/weis/VratonjicFBH11}, and the diversity in the 
Application Program Interface (API) exposed by the existing TLS/SSL libraries 
was proven a further source of security 
issues~\cite{DBLP:conf/ccs/GeorgievIJABS12}.
The latest among the reported issues on X.$509$ validation shows that, due
to a misinterpretation issue of the encoding, it was effectively possible to 
get certificates with forged signatures 
accepted~\cite{BERserk-Sketch,BERserk-NSSattack}.

An interesting point to be noted is that all the aforementioned issues do not
stem from a cryptographic vulnerability of the employed primitives, but rather
from a non systematic approach to the syntactic recognition of the certificate.
Indeed, mainly due to the high complexity of the data format, no methodical 
approach at content format recognition and syntactic verification, i.e., 
\textrm{parsing}, has been either proposed or employed in the use of existing
 X.$509$ digital certificates.   
All the existing available implementations dealing with X.$509$ certificates
employ ad-hoc handcrafted code to parse the certificate contents, in turn 
resulting in software artifacts which are difficult to test for correctness.
A practical validation of such issue is reported 
in~\cite{DBLP:conf/sp/BrubakerJRKS14} where the authors employed a tool to 
generate pseudo-random X.$509$ certificates obtained by splicing a set of valid 
ones and inserting intentional errors.

The results of testing common SSL/TLS implementations against such datasets, 
instead of purely random inputs alone, helped to uncover a significant number
of incorrect recognition issues.
However, such an approach does not provide a constructive guarantee
that a parsing strategy for X.$509$ certificates is indeed sound in its action.

A classical and sound approach to the syntactic validation of a given input 
format is the one regarding it as the problem of parsing a language specified 
by a given grammar.
Such an approach allows to provide a synthetic description of the format to be
recognized from which an implementation of the actual parser can be 
automatically derived, minimizing the implementation effort and providing
guarantees on the correctness of the recognition action.
Despite the aforementioned advantages of tackling the issue by means of a 
language theoretic approach, such a strategy has never been employed 
to parse existing X.$509$ certificates.
We ascribe this both to the complexity of the X.$509$ standard when considered
as a language over an arbitrary byte alphabet and to the fact that it was 
pointed out in~\cite{DBLP:conf/fc/KaminskyPS10} that X.$509$ is at least 
{\em context sensitive}, a feature preventing the use of most parser 
generation techniques. An interesting proposal for a new regular 
format is presented in~\cite{bmp2018itng}.
{We report that a language theoretic analysis similar to the one we 
present for X.$509$ was performed on the OpenPGP message format 
in~\cite{bmp2017fcst}.

\compactpar{Contributions.}
 In this work, we tackled the task of analyzing the standard specification of 
X.$509$ certificates from a language theoretic standpoint, highlighting which
ones of its features present a hindrance for a decidable and effective parsing.
After performing a critical analysis of the X.$509$ standard, which is specified
as a combination of Abstract Syntax Notation $1$ 
(ASN.$1$, see ITU-R X.$680$)~\cite{Xseries} data type descriptions and natural 
language defined constraints, we designed a predicated grammar defining a 
language which contains all the X.$509$ certificates employing standardized 
algorithms and smaller than $4$ GiB.

We automatically derive a parser for the said X.$509$ language from the 
aforementioned grammar specification, employing the ANTLR parser 
generator~\cite{DBLP:conf/pldi/ParrF11}.
We validate the parsing capability of our parser analyzing $11$M X.$509$ 
certificates obtained as a survey of servers running the HyperText Transfer 
Protocol over Secure Socket Layer (HTTPS) on the entire space of the 
Internet Protocol version $4$ (IPv$4$), showing that $\approx$$21.5$\% of 
the X.$509$ certificates are syntactically incorrect.

We compare our syntactic recognition capabilities against the certificate 
validation performed by $7$ widely employed TLS libraries, and show how a 
significant number of X.$509$ certificates which are 
syntactically invalid are not recognized as such by all of the current 
TLS libraries.
We validate the criticality of the detected flaws exploiting one of them to lead
a successful impersonation attack against the OpenSSL and BoringSSL libraries.
%
%
\section{Preliminaries}\label{sec:background}
In this section we report the notions needed to describe a digital certificate 
compliant with the X.$509$ standard in a form amenable to an automated parser 
generation.
In particular, we recall both the key properties of formal languages and the 
efficiency of their recognizers.  
Subsequently, since the X.$509$ specification is given employing the  
Abstract Syntax Notation $1$ (ASN.$1$) meta-syntax, we provide a bridge 
among its notation and the usual formal grammar one. 

\subsection{Parser Generation}
Given a finite set of symbols \term, known as \emph{alphabet}, a language 
$\mathbf{L}$ is defined as a set of elements, named \emph{sentences},
obtained by concatenating zero or more symbol of \term. 
Conventionally, the \emph{empty sentence} is denoted as $\varepsilon$, 
while a portion of a sentence is known as a \emph{factor}.

Given a language it is also possible to describe it via a generative formalism, 
i.e., a \emph{grammar}. 
A grammar is a quadruple $G : (\mathbf{V}_t, \mathbf{V}_n, \mathbf{P}, S)$, 
with $\mathbf{V}_t$ the alphabet of the generated language, 
$\mathbf{V}_n$ a finite set of symbols named \emph{nonterminals}, 
$\mathbf{P}$ a set of {\em productions} defined as pairs of strings obtained by 
concatenating elements of $\mathbf{V}$$=$$\mathbf{V}_t$$\cup$$\mathbf{V}_n$, 
and $S$$\in$$\mathbf{V}_n$ the {\em axiom} of the grammar.
The productions are denoted as usual as $\alpha$$\to$$\beta$; $\alpha$$\in$$
\nterm^+$, $\beta$$\in$$\mathbf{V}^*$, where $\mathbf{V}^+$ represents the 
closure with respect to concatenation over $\mathbf{V}$, and 
$\mathbf{V}^*$$=$$\mathbf{V}^+$$\cup$$\{\varepsilon\}$.
A sentence of a language $\mathbf{L}$ is generated by the grammar $G$ if it 
is possible to obtain it starting from a string $u$$\in$$\mathbf{V}^*$ made of 
the grammar axiom alone (i.e., the nonterminal $S$), and iteratively 
substituting a portion of $u$ appearing as the left hand side of a production 
in $\mathbf{P}$ with the right hand side of the said production, until the 
result is made only of elements of \term.
The grammar is said to be \emph{ambiguous}, if it is possible to generate a 
sentence of the language through two different sequences of substitutions.
Languages are classified according to the  
Chomsky hierarchy~\cite{Harrison:1978:IFL:578595}, depending on which
constraints are holding on the productions of the grammar generating them. 
Language families generated from grammars with stricter constraints are  
included in all the families with weaker constraints.
In the following, we describe the language families starting from the least 
constrained one. 
Grammars with no restrictions on the left and right sides of its productions are 
denoted as {\em unrestricted} grammars, and generate \emph{context sensitive 
with erasure} (CS-E) languages. Determining if a generic string over \term\ 
belongs to a CS-E language is not decidable.

Restricting the productions of a grammar so that the sequence of substitutions 
does not allow the erasure of symbols yields \emph{context sensitive} (CS) 
grammars, which generate CS languages.
Determining if a generic string over \term\ belongs to a CS language 
is decidable, i.e., it is always possible to state whether a string over \term\ 
can be generated by a CS grammar. 
Restricting the productions of a grammar to have a single element of \nterm\ on
their left hand side yields the so-called \emph{context-free} (CF) grammars, 
which generate CF languages. 
Deterministic CF (DCF) languages can be recognized by a Deterministic PushDown
Automaton (DPDA) and their generating grammars constitute a proper subset of 
non ambiguous CF grammars. 
Finally, restricting the productions of a grammar either to the ones of the form
$\{A$$\to$$a, A$$\to$$aA\}$ or to the ones of the form 
$\{A$$\to$$a, A$$\to$$Aa\}$, where $A$$\in$\nterm, 
$a$$\in$$\term$$\cup$$\{\varepsilon \}$, yields \emph{right-} or 
\emph{left-linear} grammars, which generate \emph{regular} (REG) languages.
We will denote such grammars as \emph{linear} whenever the direction of the 
recursion in their productions is not fundamental.

Given a grammar $G$ generating the language $\mathbf{L}$ and a generic terminal 
string $w$$\in$$\term^*$, the process of {\em parsing} consists both in 
determining if $w$$\in$$\mathbf{L}$ and in computing the sequence(s) of 
applications of the productions of $G$ which generate(s) $w$. 
The sequence(s) of productions are depicted as a tree, 
called either {\em parse tree} or {\em abstract syntax tree}, with the axiom 
of the grammar being the root and the terminal symbols being the leaves.
There exists a straightforward parsing algorithm of a CS language
which is complete and correct, but its running time is exponential in the length
of the input. As a consequence, a number of higher efficiency automated parser
generation techniques were defined for subsets of the CS grammar family. 
CF grammars are the mainstay of parser generation as it is possible to 
automatically generate a recognizer automaton for any language generated by 
them.
In particular, given a generic DCF grammar it is possible to automatically 
generate a deterministic parser for the strings of the corresponding 
language~\cite{Harrison:1978:IFL:578595}; such parser enjoys worst-case 
linear space and time requirements in the length of input. 
Linear grammars are optimal from a parsing standpoint, as it is 
possible to derive from them a parser which runs with constant 
memory and linear time requirements in the length of input.
Moreover, a parser for a REG language can be generated  
starting from the definition of a {\em regular expression}~\cite{BerryS86}. 

We aim at obtaining a grammar definition for the X.$509$ language
such that the parsing is decidable. 
In addition, we require that such grammar is expressed
as much as possible in terms of right-linear productions so that the parser will 
only require a constant amount of memory (besides the one to contain the input). 
In the remainder of this work, the terminal alphabet $\mathbf{V}_t$ will be 
the set of $256$ values which can be taken by a byte, unless pointed out 
otherwise.
Each terminal symbol will be denoted by two hexadecimal digits 
(e.g., the byte taking the decimal value $42$ will be denoted as \texttt{2A}).
We will also employ the Extended Backus-Naur Form (EBNF) to write the 
right hand sides of grammar productions.

\subsection{ASN.1 Overview}\label{subsec:asn1}
The ASN.$1$ meta-syntax was introduced by the ITU to specify structured data 
as abstract data types (ADTs), and it is described in the set of ITU 
Recommendations (ITU-R) X.$680$-X.$683$~\cite{Xseries}.
The encoding schemes for the ADTs are specified in 
ITU-R X.$690$--X.$696$~\cite{Xseries}. 
The ASN.$1$ provides the means to express both the syntax of the ADT at hand, 
in a form akin to a grammar, and some semantic constraints
concerning the values taken by an instance of such an ADT.
As our systematic parsing strategy will rely on a parser generated from a 
grammar specification for the X.$509$ ADT, we first provide a mapping between 
the syntactic-structure-specifying keywords of ASN.$1$ and the corresponding 
productions of an EBNF grammar.
Subsequently, we highlight how the remaining, non purely syntactic features of 
ASN.$1$ act as constraints on the language of the instances of the described ADT 
in terms of semantics and describe the meta-language facilities exposed to ease 
the definition of a non ambiguous specification.

\compactpar{Syntactic Elements.} 
 An ASN.$1$ ADT can be regarded as a construct equivalent to a single EBNF 
grammar generating all the possible concrete data type instances as 
its language.
The user-defined name of the ADT corresponds to the axiom of the EBNF grammar, 
while the productions are represented as structured data definitions with 
the \texttt{::=} operator separating the left and right hand side, in lieu of 
the common $\to$.

The structure of an ADT may either be a single element, in case the type is
\emph{primitive}, or a composition of other types employing ASN.$1$ constructs 
in case it is \emph{constructed}.
\begin{table}[!t]
\centering
\caption{Equivalence between the notation of the ASN.$1$ right hand side of a 
constructed type definition and the right hand side of an EBNF grammar 
production, with a set of sample ASN.$1$ types \texttt{u,v,x} matching
the identically named strings $u,v,x$$\in$$(\term$$\cup$$\nterm)^*$ in EBNF}
\label{tab:constructed types}
{\small
\begin{tabular}{l|l}
\toprule
\multicolumn{1}{c|}{ASN.$1$ right hand sides}                                    
&
\multicolumn{1}{c}{EBNF right hand sides}\\
\midrule                                                                    
{\small \texttt{SEQUENCE \{u,v,x\}          }} & $uvx$                     \\
{\small \texttt{CHOICE \{u,v,x\}            }} & $u | v | x$               \\
{\small \texttt{u OPTIONAL                  }}       & $u | \varepsilon$   \\
{\small \texttt{SET \{u,v,x\}               }} & $uvx|uxv|vux|vxu|xuv|xvu$ \\
{\small \texttt{SEQUENCE OF u               }} & $u^*$                   \\
{\small \texttt{SET OF u                    }} & $u^*$                   \\[2pt]
{\small \texttt{SEQUENCE SIZE(1 .. MAX) OF u}} & $u^+$                   \\
{\small \texttt{SET SIZE(1 .. MAX) OF u     }} & $u^+$                   \\
{\small \texttt{u(2 .. N)                 }} & $u^2 | u^3 | \ldots | u^N$\\[2pt]
{\small \texttt{ANY                       }} & Arbitrary definition      \\
{\small \texttt{ANY DEFINED BY u          }} & Arbitrary ASN.$1$ definition\\
\bottomrule
\end{tabular}
}
\end{table}

Primitive types in ASN.$1$ represent terminal rules of an EBNF grammar, 
i.e., rules where $\alpha$$\to$$\beta$, $\alpha$$\in$$\nterm^+, 
\beta$$\in$$\term^*$. 
The right hand side of a primitive type definition is described by a single
line ended by a specific keyword 
(e.g., \texttt{INTEGER}, \texttt{BOOLEAN,} \texttt{OBJECT IDENTIFIER}, 
\texttt{OCTET STRING}, \texttt{BIT STRING}), specifying completely its nature.
An \texttt{OBJECT IDENTIFIER} is a sequence of decimal numbers separated 
by dots, of which a single decimal number is known as an \emph{arc}.

Constructed types may either have a single user-defined name appearing on the
right hand side of their definition, in which case they act as the copy rules
of an EBNF grammar (i.e., rules where $\alpha$$\to$$\beta,
\alpha,\beta$$\in$$\nterm$), or an arbitrary ASN.$1$ syntactic construct may
be used.
The only exception to the aforementioned constructed type definition is
the possibility of turning a primitive \texttt{OCTET STRING} or \texttt{BIT 
STRING} type into a constructed one through appending the keyword 
\texttt{CONTAINING} to it, followed by the description of its contents.
Deriving a constructed type from an \texttt{OCTET STRING} forces its value in an 
ADT instance to be byte-aligned, and allows the designer to enforce the choice 
of the encoding rules to be employed for it with the \texttt{ENCODED BY} 
keywords followed by the encoding name (see ITU-R X.$682$)~\cite{Xseries}.
Table~\ref{tab:constructed types} shows a mapping between the ASN.$1$ 
structures appearing on the right hand side of the data types which are  
considered in this work and their matching EBNF notation, expressed 
with a set of sample ASN.$1$ types \texttt{u,v,x} and identically named 
strings $u,v,x$$\in$$(\term$$\cup$$\nterm)^*$ in EBNF. 

The ASN.$1$ notation specifies the common concatenation and alternative
choice via the \texttt{SEQUENCE} and \texttt{CHOICE} keywords, respectively, 
while the syntactic constraint indicated by the \texttt{SET} keyword 
mandates that a set of ADTs may appear in any order, without repetitions. 
Postfixing an ADT appearing in a right hand side of a definition with the 
\texttt{OPTIONAL} keyword allows it to be either missing or present only once.
Given an ASN.$1$ ADT \texttt{u}, the syntactic constraints imposed by the 
\texttt{SEQUENCE OF} and \texttt{SET OF} constructs, indicate a concatenation of
zero or more instances of \texttt{u}, matching the star operator in EBNF.
The ternary range operator in ASN.$1$, having the syntax 
\texttt{t(low .. high)}, where \texttt{t} is an ADT and \texttt{low,high} 
the range boundaries, is employed with two purposes: specifying the range of 
possible values of the instances of the primitive ADT to which they are 
appended, or indicating the concatenation of any number 
\texttt{low}$<$$n$$<$\texttt{high} of instances of the user-defined ADT 
preceding them.
The bounds of a range operator are either constant values or the keywords 
\texttt{MIN} and \texttt{MAX}, which indicate that the minimum (resp., maximum) 
of the given range, is interpreted as the smallest (resp., greatest) value which 
can be taken by the ADT on their left.
ASN.$1$ allows to specify \emph{size constraints} through the use of the
keyword \texttt{SIZE}, followed by a range of allowed sizes.
A common idiom in ASN.$1$ ADT declarations is to employ \texttt{MAX} as an 
upper bound for \texttt{SIZE}s to indicate that there is no upper bound on 
the size.
As a consequence, the two common ASN.$1$ \texttt{SEQUENCE SIZE(1 .. MAX) OF} 
and \texttt{SET SIZE(1 .. MAX) OF} idioms in Table~\ref{tab:constructed types} 
are representing a concatenation of one or more instances of the involved ADT  
\texttt{u}, corresponding to the EBNF cross construct. 
Finally, the ASN.$1$ keyword \texttt{ANY} allows to delegate the definition of 
the structure of a given ADT to another document, potentially not expressed in 
ASN.$1$.
Specifying further the effect of the \texttt{ANY} construct with the 
\texttt{DEFINED BY} keywords, followed by the name of an ADT, enforces the
fact that the specification should be expressed in ASN.$1$.

\compactpar{Semantic Elements and Disambiguation Constructs.} 
 ASN.$1$ allows to describe semantic information 
on instances of an ADT either specifying a constant value for a given 
{\em primitive type} element, appending such value between round brackets on 
the line where the type appears, or specifying a so-called \texttt{DEFAULT} 
value.
Appending the \texttt{DEFAULT} keyword allows to indicate that, in case an 
element is missing in an instance of the ADT, the recognizer should assume its 
presence nonetheless, and assign the semantic value present in the ASN.$1$ 
specification to it.
The expressive power of the ASN.$1$ allows the designer to specify an ADT 
corresponding to an \emph{inherently ambiguous} language, i.e., a language for 
which no unambiguous grammar exists.
An illustrative example of such a case is the following ADT \texttt{t}:
\begin{center}
{\small
\begin{lstlisting}[language=asn.1]
 t ::= CHOICE {
  s1 SEQUENCE{u(i..i),v(i..i),x(j..j)},
  s2 SEQUENCE{u(j..j),v(i..i),x(i..i)}}
 i ::= INTEGER 1 .. MAX
 j ::= INTEGER 1 .. MAX
\end{lstlisting}
}
\end{center}
\noindent which has its instances belonging to the intrinsically ambiguous 
language $\mathbf{L}$$=$$\{u^iv^ix^j\ $ $\vee\ u^jv^ix^i \ \mathrm{s.t.}$ $ 
u,v,x$$\in$$\term,\ i$$\geq$$1, j$$\geq$$1 \}$.
To provide a convenient way to cope with ambiguities, ASN.$1$ introduces 
the so-called user-defined \emph{tag} elements.
A \emph{tag} is a syntactic element, denoted as a decimal number enclosed in 
square brackets, which is prefixed to an ADT appearing in the right hand side of 
a data description. 
Proper use of {\em tag}s minimally alters the language of accepted ADT 
instances, while effectively curbing ambiguities.
The inherent ambiguity of the sentences of the language of the previously shown 
ADT \texttt{t} can be eliminated adding two {\em tag}s to its description: 
\vspace{-0.25cm}
\begin{center}
{\small
\begin{lstlisting}[language=asn.1]
 t ::= CHOICE {
  s1 SEQUENCE{ [0] u(i..i),v(i..i),x(j..j)},
  s2 SEQUENCE{ [1] u(j..j),v(i..i),x(i..i)}}
 i ::= INTEGER 1 .. MAX
 j ::= INTEGER 1 .. MAX
\end{lstlisting}
}
\end{center}
\vspace{-0.35cm}
\noindent It is worth noting that there is no algorithmic way to check whether a
given set of \emph{tag}s introduced in an ASN.$1$ specification is enough to 
suppress all parsing ambiguities, as the expressive power of ASN.$1$ is 
sufficient to define a generic CF grammar, for which determining if it is 
ambiguous is a well known undecidable problem~\cite{Harrison:1978:IFL:578595}.

\subsection{Distinguished Encoding Rules}\label{subsec:der}
The encoding rules for an ASN.$1$ data type instance 
(see ITU-R X.$690$--X.$696$~\cite{Xseries}) define several formats portable 
across different architectures by mandating bit and byte value conventions and 
ordering of the encoded contents. 
We tackle the Distinguished Encoding Rules (DER), as a   
{\em constructed typed} field in the X.$509$ standard certificate ADT requires 
its DER encoding via the \texttt{ENCODED BY} keyword. 
Consequentially, DER is the encoding rule employed in the 
overwhelming majority of X.$509$ certificate instances found.
DER encoded material may be further mapped to the fully printable 
\texttt{base64} encoding, resulting in the so-called Privacy-enhanced Electronic 
Mail (PEM) format~\cite{RFC1421}. 
The DER encoding strategy represents an ASN.$1$ ADT instance as a stream of 
bytes which is logically split up into three fields: 
\emph{identifier octets}, \emph{length octets}, \emph{contents octets}.

The \emph{identifier octets} field is employed to encode the ASN.$1$ tag value
and whether the ADT instance at hand is a primitive or a constructed 
ASN.$1$ type.
The tag value may be either the disambiguating user-defined one present in the 
ASN.$1$ ADT definition, or a so-called \emph{universal tag} assigned by the DER 
standard to all ASN.$1$ primitive types and to the 
\texttt{SEQUENCE, SEQUENCE OF, SET, SET OF} constructs.
If a user-defined tag is present, its encoding is stored in the 
\emph{identifier octets} field, while the encoding of the tagged ADT is stored 
in the \emph{contents octets} field.
To provide a more succinct encoding, ASN.$1$ allows to specify that a given 
user-defined tag is \texttt{IMPLICIT}, i.e., that it should replace the tag of 
the tagged ADT in the encoding in the \emph{identifier octets} field.
On the other hand, the \texttt{EXPLICIT} keyword states that a user-defined 
tag should be encoded according to the default behavior.
\newbox\listboxa
\begin{lrbox}{\listboxa}
	\begin{lstlisting}[language=asn.1]
	Certificate ::= SEQUENCE {
	tbsCert            TBSCertificate,
	signatureAlgorithm AlgorithmIdentifier, 
	signatureValue     BIT STRING }
	
	TBSCertificate ::= SEQUENCE {
	version [0] EXPLICIT Version DEFAULT v1,
	serialNumber         CertSerialNumber,
	signature            AlgorithmIdentifier,
	issuer               Name,
	validity             Validity,
	subject              Name,
	subjectPublicKeyInfo SubjectPubKeyInfo,
	issuerUniqueID  [1] IMPLICIT UniqueId   OPTIONAL,
	subjectUniqueID [2] IMPLICIT UniqueId   OPTIONAL,
	extensions      [3] EXPLICIT Extensions OPTIONAL }
	
	AlgorithmIdentifier ::= SEQUENCE {
	algorithm  OBJECT IDENTIFIER,
	parameters ANY DEFINED BY AlgorithmP OPTIONAL}
	
	SubjectPublicKeyInfo  ::=  SEQUENCE  {
	algorithm         AlgorithmIdentifier,
	subjectPublicKey  BIT STRING }
	\end{lstlisting}
\end{lrbox}

\begin{figure}[!h]
	\centering
	\usebox\listboxa
	\caption{Portion of the description of the X.$509$ \texttt{Certificate} ADT 
		and its fields\label{lst:x509main}}
\end{figure}

The \emph{contents octets} field contains the actual encoding of the 
ADT instance at hand, while the {\em length octet} one stores the size of the 
content field as a number of bytes. 
The number of bytes constituting the \emph{length octets} field varies in the 
$1$ to $126$ range. The encoding conventions for the length value are stated in 
the ITU-R X.$690$~\cite{Xseries}. 
A short form and a long form for the encoding of a length value are possible. 
The short form mandates the encoding of the length field as a single octet in 
which the most significant bit is $0$ and the remaining ones encode the size of 
the {\em contents octet} field from $0$ to $127$ bytes.  
The long form consists of one initial octet followed by one or more subsequent 
octets containing the actual value of the {\em length octet} field.
In the initial octet, the most significant bit is $1$, while the remaining ones 
encode the number of length field octets that follow as an integer varying from 
$1$ to $126$. 
Thus, the number of bytes for the {\em contents octet} field is at most 
$2^{126 \cdot 8}$.
The {\em length octet} field value equal to $128$ encoded as a single byte 
is forbidden in DER, while in other standard encoding rules it is reserved to 
indicate that an indefinite number of bytes will follow. 
It is worth noting that the requirement to validate correctly the information in 
the {\em length octets} field against the actual length of the {\em contents 
octets} one can be done via a regular language matcher, as the lengths have an 
upper bound.
Nevertheless, the corresponding minimal finite state automaton (FSA) has a 
number of states which cannot be represented (see Section~\ref{sec:grammar}).

\subsection{Description of X.509 Certificate Structure}

The structure of an X.$509$ certificate is described as an ASN.$1$ ADT in the 
ITU-R X.$509$~\cite{X509}, and in the RFC $4210$~\cite{RFC4210} and its 
complements~\cite{RFC3279,RFC4055,RFC4491,RFC5280,RFC5480,RFC5758}. 
Its structure evolved over time, as its first version dates back 
to $1988$. In this section, we provide a synthetic description of its contents.

\compactpar{Certificate Abstract Data Type.} 
Figure~\ref{lst:x509main} reports a shortened version of the X.$509$ standard, 
defining the main \texttt{Certificate} ADT. 
In the figure, field names start with a lowercase letter, while ASN.$1$ 
user-defined ADT names start with capital letters.
The \texttt{Certificate} ADT is a concatenation of three fields: the material to
be signed typed as \texttt{TBSCertificate}, the identification data for the 
signature algorithm typed as \texttt{AlgorithmIdentifier}, and a 
\texttt{BIT STRING} field containing the actual signature value.

\compactpar{TBSCertificate Abstract Data Type.} 
Considering the contents of the \texttt{TBSCertificate} ADT, 
the first two fields contain a version number typed as \texttt{Version} (a value 
constrained \texttt{INTEGER}), and an integer value which must be unique 
among all the certificates signed by the same certification authority (CA), 
which is typed as \texttt{CertSerialNumber}.
The third field, typed as an \texttt{AlgorithmIdentifier}, contains the 
information to uniquely identify the cryptographic primitive employed to sign 
the certificate.

The \texttt{AlgorithmIdentifier} ADT is a concatenation of two fields:
an ASN.$1$ \texttt{OBJECT IDENTIFIER} (OID) typed field
\texttt{algorithm} and
an optional \texttt{parameters} field, typed as 
\texttt{ANY\,DEFINED\,BY\,AlgorithmP}. 
The OID value allows to uniquely label the signature algorithm to allow 
the description of the structure of its parameters, done 
in~\cite{RFC3279,RFC4055,RFC4491,RFC5480,RFC5758} for a set of standardized 
signature algorithms.
The \texttt{issuer}, \texttt{validity} and \texttt{subject} fields 
contain information on the CA issuing the certificate and the subject to whom it
has been issued, together with the validity time interval of the certificate.
\texttt{issuer} and \texttt{subject} fields are typed as a \texttt{Name} ADT: 
a \texttt{SEQUENCE OF} of \texttt{SET OF} structures containing 
a concatenation of two fields typed as OID and \texttt{ANY}, respectively.
Despite the quite baroque definition, the \texttt{Name} ADT is indeed employed
to represent a list of names for both the issuer and the subject which are 
typically expressed as printable strings prefixed with a standardized OID 
value stating their meaning (e.g., organization, state).

The \texttt{subjectPublicKeyInfo} field provides both the public key binded 
to the subject identity, and information on the employed cryptographic 
primitive in the form of a \texttt{BIT STRING} and an 
\texttt{AlgorithmIdentifier} typed field, respectively.
Following the \texttt{subjectPublicKeyInfo} field, the \texttt{TBSCertificate} 
ADT includes two deprecated extra optional fields, 
containing further information about the issuer and the subject.
These fields are tagged with tags \texttt{[1]} and \texttt{[2]} 
respectively, preventing a possible parsing ambiguity arising from 
only one of them being present.\\
\newbox\listboxb
\begin{lrbox}{\listboxb}
	\begin{lstlisting}[language=asn.1]
	Extensions ::= SEQUENCE SIZE(1..MAX) OF Extension
	
	Extension ::= SEQUENCE {
	extnID    OBJECT IDENTIFIER,
	critical  BOOLEAN DEFAULT FALSE,
	extnValue OCTET STRING CONTAINING ...  ENCODED BY der }
	\end{lstlisting}
\end{lrbox}
\begin{figure}[!t]
\centering
\usebox\listboxb
\caption{X.$509$ ASN.$1$ description reporting
the definition of the last field of the \texttt{TBSCertificate} 
ADT\label{lst:x509ext}}
\end{figure}

\compactpar{Extension Abstract Data Type.} 
The \texttt{extensions} field concludes the definition of the 
\texttt{TBSCertificate} ADT.
Most of the information of modern certificates is contained in it, and its 
presence is mandatory in the current version (v$3$) of X.$509$ certificates. 
As reported in Figure~\ref{lst:x509ext} the \texttt{Extensions} ADT is a 
sequence of one or more \texttt{Extension} typed fields. 
Each \texttt{Extension} ADT is composed of an OID typed field identifying it 
unambiguously, and a \texttt{critical} field typed as \texttt{BOOLEAN} 
indicating, if {\tt True}, that the certificate validation should fail in case 
the application either does not recognize or cannot process the information 
contained in the subsequent \texttt{extnValue} field.
An example of information contained in the \texttt{extnValue} field is the
so-called {\tt KeyUsage}, i.e., information stating which is the legitimate 
purpose of the subject public key in the certificate at hand 
(e.g., signature validation or encryption). 
%
%
\section{Analysis of the X.509 Certificate Language}\label{sec:x509}
In this section, we analyze the X.$509$ certificate structure
from a language theoretic standpoint. In particular, we highlight the
portions of the certificate which hinder and harden the design of
a grammar amenable to automatic parsing generation algorithms. These 
issues will be tackled in the next section, in order to achieve our 
goal of obtaining an effective and decidable parser for X.$509$ digital
certificates.

\compactpar{Undecidability of The Parsing.}
Some portions of an X.$509$ certificate requires an unrestricted grammar
in order to be generated, in turn implying the parsing undecidability. 
First, consider the {\tt signature} field in the 
{\tt TBSCertificate} ADT, which is typed as an {\tt AlgorithmIdentifier}
ADT. This is composed by an OID, identifying the signature algorithm, and
by optional {\tt parameters}, whose structure depends on the OID. The 
complete freedom in the specification of the {\tt parameters} 
field via the \texttt{ANY} keyword of ASN.$1$ meta-syntax allows the description 
to be arbitrarily complex, in particular allowing natural language constraints
to be specified on eventual fields of \texttt{parameters}.
This, in turn, requires an \emph{unrestricted} grammar to generate the 
language and results in the undecidability of the certificate parsing.
A similar issue arises in {\tt extnValue} field of the {\tt Extension} ADT.
This field is typed as an \texttt{OCTET STRING} and is turned into a 
constructed type by the keyword {\tt CONTAINING} 
(see Section~{\ref{subsec:asn1}}). 
The specific structure of the \texttt{extnValue} field  depends on the 
value of the {\tt extnID} field, which is an OID identifying the extension type. 
The standard allows the definition of custom extensions, which may specify an 
arbitrary content for it~\cite{X509}. Such a lack of constraints allows the
checks on this content to be arbitrary complex, i.e., representable by an
unrestricted grammar.

\compactpar{Context-Sensitiveness.}
Some portions of X.$509$ certificate structure cannot be generated by a 
context-free grammar, in turn implying that for these portions an efficient 
automatically generated parser cannot be derived. 
The context-sensitiveness is introduced by the same kind of constraint in $2$ 
different portions of the certificates. This constraint is the
need to check repetitions of arbitrarily long strings. 
First, consider the {\tt signatureAlgorithm} field in {\tt Certificate} ADT and 
the {\tt signature} field in {\tt TBSCertificate ADT}, both typed as 
{\tt AlgorithmIdentifier}. 
While the latter is found in the portion of the certificate which is signed, 
the former is not. 
Therefore, it is expected that these $2$ fields have the same content, in
turn requiring a Context-Sensitive (CS) recognizer to check this equality. 
Similarly, some portions of the certificate must be present if the certificate 
is self-issued, which means that the issuer and the subject are the same entity. 
Recall that both {\tt issuer} and {\tt subject} are typed as {\tt Name} ADT, 
which is an arbitrarily long sequence of names, which are generally printable 
strings, referring to the issuer or subject entity, coupled with an OID 
identifying their meaning. 
Verifying if a certificate is self-issued or not requires to match the content 
of {\tt issuer} and {\tt subject} fields, which are arbitrarily long string 
of bytes.

\compactpar{Parsing Ambiguities and Inconsistencies.}
Due to looseness in standard specification of some fields, there are also 
parsing ambiguities in the format. 
First, consider the {\tt signatureValue} field in the {\tt Certificate} ADT. 
Despite the X.$509$ standard is typing this field as a primitive 
\texttt{BIT STRING}, some standardized signature algorithms require it to be a 
constructed field (e.g., the DSA and ECDSA cryptographic 
primitives~\cite{RFC5480,RFC5758}), in turn giving way to an 
ambiguity in its interpretation.
The same kind of issue arises in {\tt SubjectPublicKeyInfo} field in 
{\tt TBSCertificate} ADT. 
Indeed, the public key field is typed as a \texttt{BIT STRING}, 
instead of a constructed ADT. These ambiguities may lead to security issues 
when parameters for a given cryptographic primitive are either misinterpreted 
as valid for another one or simply parsed 
incorrectly~\cite{BERserk-Sketch,weakdh15}.
A further issue is introduced by the {\tt extnValue} field in {\tt Extension} 
ADT, which is typed as {\tt OCTET STRING} but contains the extension data, 
which is usually a constructed ADT. Nevertheless, DER forbids the encoding of 
\texttt{OCTET STRINGS} as constructed types 
(see ITU-R X.$690$)~\cite{Xseries}, in turn forcing an inconsistency
in the way \texttt{OCTET STRINGS} containing the value of an \texttt{extnValue} 
field should  be treated during parsing (due to the {\tt CONTAINING} keyword in 
the X.$509$ specification). 
We note that a less problematic definition of the field would have involved
a dedicated constructed ADT for the \texttt{extnValue} field, which should have 
had its structure specified according to the value of the \texttt{extnID} field.

\compactpar{Unmanageable Number of Rules.}
Recall that the {\tt Extensions} ADT is a sequence of {\tt Extension} ADTs. 
Such a sequence of ADTs has no constraint on their order, as they all share the 
same ADT. However, a  constraint expressed in natural language in the X.$509$ 
standard mandates that $2$ \texttt{Extension} typed field instances with the 
same \texttt{extnID} field cannot appear, in turn providing a concrete hindrance 
to its representation in grammar form. 
Indeed, an exponential number of productions would be required to generate 
unique \texttt{Extension} instances in any possible order. 
Considering that the number of X.$509$ standardized extensions is 
$17$~\cite{X509}, the required grammar would have at least $2^{17}$ 
productions, which is hardly manageable by a designer.

\subsection{Summing Up Hindrances to X.509 Parsing}\label{subsec:hindrances}
In the following statements we provide a synthesis of the undecidability, 
context sensitiveness and ambiguity problems arising from the X.$509$ 
certificate standard and recall the issue about the size of the grammar 
representation reported in Section~\ref{subsec:der}.

\begin{statement}{\sc(Undecidability Context Sensitiveness and Ambiguity)}
 The following issues should be tackled lest no unambiguous, correct parsing
 of X.$509$ is possible:
\begin{description}
\item[$\phantom{ii}i$)] Cope with the undecidability arising from the potential 
definition of arbitrary structures for both the \texttt{AlgorithmP} ADT 
in the right-hand side of the \texttt{AlgorithmIdentifier} ADT 
and the value of the \texttt{extnValue} field in the right-hand side of 
\texttt{Extension} ADT, without requiring an \emph{unrestricted} 
grammar to specify them.
\item[$\phantom{i}ii$)] 
Cope with the context sensitiveness of the X.$509$ grammar, arising from the 
equality checks required on the \texttt{signature} field of the 
\texttt{TBSCertificate} ADT against the \texttt{signatureAlgorithm} field 
of \texttt{Certificate} ADT, and the need to match arbitrarily long 
\texttt{subject} and \texttt{issuer} fields of the \texttt{TBSCertificate} ADT.
\item[$iii$)] 
Cope with the ambiguity introduced by the definition of the \texttt{extnValue}
field as a constructed \texttt{OCTET STRING}, while its encoding is forced
to mark it as a primitive one. Similarly cope with the issue of the 
\texttt{signatureValue} field and the \texttt{subjectPublicKey} field where
information represented as a constructed data type is encoded as a 
primitive \texttt{BIT STRING}. 
\end{description}
\end{statement}

\noindent In addition to the aforementioned issues, some of the portions of the 
X.$509$ language which can be defined with a linear grammar 
are still unpractical to be fully specified.

\begin{statement}{\sc (Unmanageable Grammar Representation)}
\label{stmt:largeFSA} 
The presence of uniqueness constraints among the instances of the 
\texttt{Extension} ADT standardized in~\cite{RFC5280} can only be expressed 
with a linear grammar having $>$$2^{17}$ productions. 
\end{statement}

\noindent We consider the aforementioned issues to be among the main causes 
of the current approach to X.$509$ parsing, which sees fully handcrafted 
implementations of the parser as an alternative to the development of an 
automatically generated one.
However, given the complexity of the X.$509$ certificate standard specification, 
and the impossibility of methodically checking the correctness of a 
human-implemented recognizer (due to context-sensitive portions of the X.$509$ 
language), the current state-of-the-art parser implementations are exhibiting a 
non-trivial amount of security critical errors in the recognition of X.$509$ 
sentences~\cite{marlinspike02,marlinspike09,DBLP:conf/fc/KaminskyPS10}.
%
%
\section{Systematic X.509 Parsing}\label{sec:grammar}
In this section we describe our approach to realize an automatically generated 
X.$509$ parser, starting from a grammar description.
The obtained parser enjoys termination guarantees, correctness and a 
worst-case computational complexity which is quadratic in the length of the 
certificate.

We strive to minimize the non-regular portion of X.$509$ language, allowing an 
efficient parser to be systematically generated.
However, since it is not possible to turn X.$509$ into a regular language 
without either significant precision loss in recognition, or substantial 
functionality loss in the parser, we retain a small CS portion of its 
description.
Such a CS portion, indeed made of two binary string equality checks and a single
boolean condition, is implemented by hand, as relying on the generic CS parsing
algorithm would entail an exponential running time even for such rather trivial 
operations.

\subsection{Coping with Undecidability, Context Sensitiveness and Ambiguity}
\label{subsec:undec}
The first step to obtain a useful grammar description of the X.$509$ standard
is to cope with its parsing undecidability problem.
We restricted the \texttt{AlgorithmIdentifier} ADT descriptions to the 
ones which are explicitly and fully described in the 
standards~\cite{RFC3279,RFC4055,RFC4491,RFC5758,RFC5480}.
As a consequence we remove the possibility of having an unrestricted syntactic
structure of their fields allowed by the ASN.$1$ keyword \texttt{ANY} in the 
standard.
This decision implies that we will recognize certificates employing only 
algorithms for which a standardized \texttt{AlgorithmIdentifier} description 
exists, which is typically a sign of wide acceptance of the cryptographic 
security margin offered by them.
Examining all the standardized \texttt{AlgorithmIdentifier} descriptions 
we note that their structures are described by regular languages.
This restriction may prevent correct recognition of valid X.$509$
certificates found in the wild. During our practical evaluation
campaign, we observed just a negligible percentage of certificates being 
affected by this restriction, as detailed in Section~\ref{sec:expres}.
We tackled the similar problem caused by the \texttt{ANY}-typed 
\texttt{extnValue} field by parsing the syntactic structure of all
the standard defined extensions~\cite{RFC5280} strictly.

However, in contrast with the decision taken for the non standard 
\texttt{AlgorithmIdentifier} structures, we accept user-defined extensions, 
considering them opaque byte sequences, and performing only a syntactic check
on the correctness of the OID field and the length of the unspecified structure
field.
Indeed, in case of a custom extension for which the Boolean \texttt{critical} 
field is set, a validation failure must be signaled by the application logic, in
case it is not able to handle the contents of the \texttt{extnValue} 
field~\cite{X509}.
Taking into account the described restrictions, the resulting X.$509$ 
specification is a CS language $\mathbf{L}_{\rm X.509-r1}$ where the 
string membership problem is decidable.

The second step in producing a grammar description of X.$509$ is to tackle its
context-sensitiveness.
To this end, we consider the language $\mathbf{L}_{\rm X.509-r1}$ as the 
intersection of three languages 
$\mathbf{L}_{\rm X.509-r1}=\mathbf{L}_{\rm AId-match} \cap 
\mathbf{L}_{\rm SI-match} \cap \mathbf{L}_{\rm CF}$, where 
$\mathbf{L}_{\rm AId-match}$ is the language containing sentences with two 
matching \texttt{AlgorithmIdentifier} type instances, 
$\mathbf{L}_{\rm SI-match}$ is the one where whenever the certificate 
\texttt{subject} and \texttt{issuer} match, the \texttt{keyIdentifier} field 
of the \texttt{AuthorityKeyIdentifier} extension ADT should be present, and 
$\mathbf{L}_{\rm CF}$ is the CF language where all the remaining 
X.$509$ linguistic constraints are holding on the sentences.
Thus, parsing $\mathbf{L}_{\rm X.509-r1}$ can be done analyzing the
candidate string (i.e., a certificate instance) with three parsers, each one 
recognizing one of the three aforementioned languages.
We choose to start with the parser $\mathcal{A}_{\rm CF}$ for 
$\mathbf{L}_{\rm CF}$, as it can be automatically generated from its grammar 
description. 
The parse tree produced by $\mathcal{A}_{\rm CF}$ is then employed by a 
handwritten string matcher, which checks the conditions required for a string 
to belong both to $\mathbf{L}_{\rm AId-match}$ and to 
$\mathbf{L}_{\rm SI-match}$.
The code auditing of such a simple handwritten string matcher can be confidently
performed by direct inspection to check its correctness. 
Indeed, these pieces of code are small and definitely simpler than a full 
hand-written parser, dramatically decreasing the complexity of the code 
auditing process. 
Therefore, spotting any vulnerabilities in these hand-written parsers is 
expected to be easier. 
Finally, we remark that this code has no influence on the automatic
generated parser for $\mathcal{A}_{\rm CF}$, since these matchings are
performed when the parse tree for $\mathcal{A}_{\rm CF}$ has already been built.
Consequently, it is not possible that a subtle vulnerability in this code 
affects the correctness of the $\mathcal{A}_{\rm CF}$ parser. 

The last step to achieve unambiguous, decidable recognition of the X.$509$ 
certificates is to cope with the ambiguity issues stemming from the primitive 
encoding forced by DER on constructed \texttt{OCTET STRING}s in the 
\texttt{extnValue} field and the \texttt{BIT STRING}s containing 
signature and public key material specified by the standard to be primitive.
Concerning the issue of the \texttt{extnValue} field, we choose to eliminate
the ambiguity relying on the fact that the standards define the contents of 
the constructed type binding it to a specific OID: this in turn allows to 
 parse unambiguously an ADT instance.
Concerning the \texttt{BIT STRING} fields, we eliminate the parsing ambiguity
recognizing them as constructed or primitive types according to the individual
cryptographic primitive needs as specified in the 
standards~\cite{RFC3279,RFC4055,RFC4491,RFC5758,RFC5480}.
Such an approach will indeed pick only a single way of interpreting the data
contained in the field, preventing security critical parsing issues such as
the ones in~\cite{weakdh15}.

Once such ambiguities are removed from $\mathbf{L}_{\rm CF}$, the resulting 
language is indeed regular and will be denoted as $\mathbf{L}_{\rm REG}$ from 
now on, while the resulting restricted X.$509$ language will be denoted as 
$\mathbf{L}_{\rm X.509-r2}$$=$$\mathbf{L}_{\rm AId-match} \cap 
\mathbf{L}_{\rm SI-match} \cap \mathbf{L}_{\rm REG}$.

\begin{lemma} The proposed parser for $\mathbf{L}_{\rm X.509-r2}$ terminates
on any input string.
\end{lemma}
\begin{proof} Since $\mathbf{L}_{\rm X.509-r2}$$=$$\mathbf{L}_{\rm AId-match} 
\cap \mathbf{L}_{\rm SI-match} 
\cap \mathbf{L}_{\rm REG}$ is the result of the intersection of two CS languages 
and a REG language, the closure property of the CS family w.r.t. set 
intersection implies that $\mathbf{L}_{\rm X.509-r2}$ is CS, a necessary
condition for the termination of its parsing process. 
To prove the termination guarantee of the proposed parser, it is sufficient to 
observe that all the three parsers, sequentially executed to recognize 
$\mathbf{L}_{\rm X.509-r2}$, have guarantees on their termination, namely: 
the parser for $\mathbf{L}_{\rm REG}$ is a FSA, while the string 
matchers act on arbitrarily large, but not infinite strings.
\end{proof}

\begin{lemma} $\mathbf{L}_{\rm X.509-r2}$ is non intrinsically ambiguous, and 
all its sentences can be parsed unambiguously.
\end{lemma}
\begin{proof} The only source of ambiguity in the original X.$509$ specification
is the one caused by interpretation issues of the X.$509$ standard, which 
are no longer present in $\mathbf{L}_{\rm X.509-r2}$.
The first stage of the proposed parser matches the sentences of 
$\mathbf{L}_{\rm X.509-r2}$$=$$((\mathbf{L}_{\rm REG} \cap 
\mathbf{L}_{\rm AId-match}) \cap \mathbf{L}_{\rm SI-match})$ by means of a 
deterministic finite state recognizer, 
which can only have a single recognition path for each one of them.
The CS recognizer stages simply discard some of the strings deemed
valid by the first stage, and thus do not introduce ambiguity.
\end{proof}

\subsection{Managing Grammar Representation and Recognizer Generation}

Having dealt with the linguistic issues which could either have prevented the
parsing (i.e., the undecidability) or made its result unusable (i.e., the 
ambiguity), we now tackle the issues concerning the practical construction of 
a grammar representation for X.$509$ and the generation of its recognizer.
Among the recognizers required to match $\mathbf{L}_{\rm X.509-r2}$$=
$$((\mathbf{L}_{\rm REG} \cap \mathbf{L}_{\rm AId-match} )\cap 
\mathbf{L}_{\rm SI-match})$ the last two suffer from no issues either in their
syntax definition or in their recognizer generation.

In order to tackle the issues coming from Statement~\ref{stmt:largeFSA} in 
Section~\ref{subsec:hindrances}, we consider $\mathbf{L}_{\rm REG}$ as the 
intersection of three regular languages: 
$\mathbf{L}_{\rm REG} = \mathbf{L}_{\rm REG-1} \cap \mathbf{L}_{\rm REG-2}\cap 
\mathbf{L}_{\rm REG-3}$.
$\mathbf{L}_{\rm REG-1}$ is the language of byte strings containing sequences 
of \texttt{Extension}s instances where each \texttt{Extension} instance appears 
at most once, $\mathbf{L}_{\rm REG-2}$ is the language of DER encoded ASN.$1$ 
fields with the correct constraints holding between the {\em length octet} field 
and the {\em contents octets} field, and $\mathbf{L}_{\rm REG-3}$ contains all 
the remaining regular constraints to obtain a sentence in $\mathbf{L}_{\rm REG}$
when considered in intersection with $\mathbf{L}_{\rm REG-1}$ and 
$\mathbf{L}_{\rm REG-2}$.

A straightforward definition of $\mathbf{L}_{\rm REG-1}$ with a linear grammar 
requires $17!$ rules, a number which can be reduced to a little over $2^{17}$ 
through careful rewriting. 
We deem this number of productions too high to be described and checked 
reliably, thus we match a relaxation of $\mathbf{L}_{\rm REG-1}$ at first, which
allows the \texttt{Extension}s to be duplicated, and perform a subsequent 
uniqueness check on the result of the parsing action.
This choice allows us to retain the description of the structure of the 
individual \texttt{Extension} instance as a purely syntactic one, and perform 
the check for duplicates on their parsed OID fields.
We implemented the duplicate checking logic as an handwritten portion of 
our parser, as we deem the code to be small and simple enough to allow effective
code auditing of it.

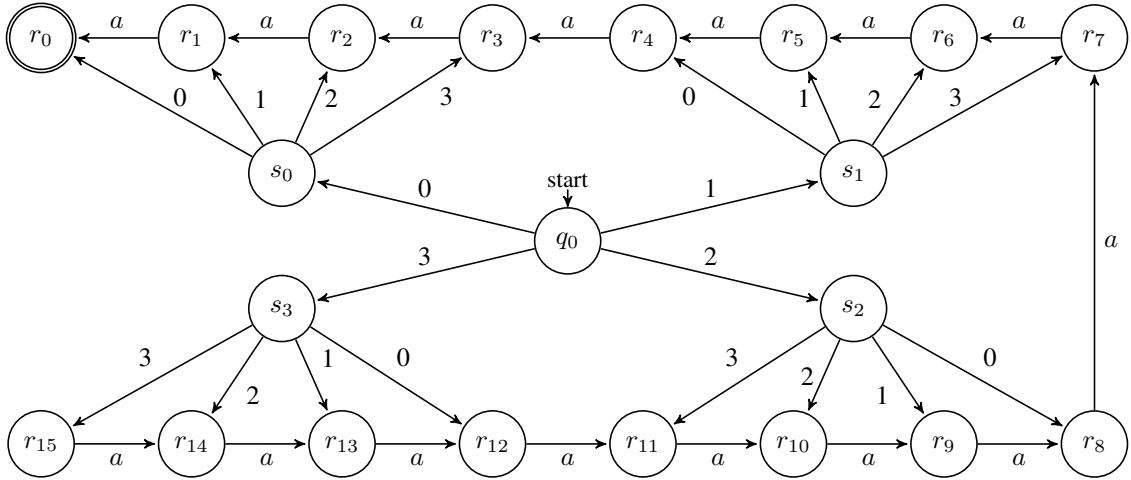
\begin{figure}[!t]
\begin{center}
\begin{tikzpicture}[->,>=stealth',shorten >=1pt,
                    auto,node distance=2cm, 
                    scale=0.6,semithick]
  \tikzstyle{every state}=[fill=none,draw=black,text=black]
  \node[state,accepting](0)                     {$r_{0 }$};
  \node[state]         (1)  [right of=0]        {$r_{1 }$};
  \node[state]         (2)  [right of=1]        {$r_{2 }$};
  \node[state]         (3)  [right of=2]        {$r_{3 }$};
  \node[state]         (4)  [right of=3]        {$r_{4 }$};
  \node[state]         (5)  [right of=4]        {$r_{5 }$}; 
  \node[state]         (6)  [right of=5]        {$r_{6 }$};
  \node[state]         (7)  [right of=6]        {$r_{7 }$};

  \node[state]         (0x) at (5.333 ,-3)       {$s_{0}$};  
  \node[state]         (1x) at (18	,-3)         {$s_{1}$};  
  \node[state]         (q0) at (11.666,-4.5)         {$q_{0}$};
  \draw[<-]            (q0) -- node[above]      {\small start} ++(0,1.2);

  \node[state]         (3x) at (5.333 ,-6)           {$s_{3}$};  
  \node[state]         (2x) at (18,-6)          {$s_{2}$};   

  \node[state]         (15) at (0,-9)           {$r_{15}$};
  \node[state]         (14) [right of=15]       {$r_{14}$};
  \node[state]         (13) [right of=14]       {$r_{13}$};
  \node[state]         (12) [right of=13]       {$r_{12}$};
  \node[state]         (11) [right of=12]       {$r_{11}$};
  \node[state]         (10) [right of=11]       {$r_{10}$}; 
  \node[state]         (9)  [right of=10]       {$r_{ 9}$};
  \node[state]         (8)  [right of=9]        {$r_{ 8}$};

  \path (1)   edge []    node [above] {$a$} (0)
        (2)   edge []    node [above] {$a$} (1)
        (3)   edge []    node [above] {$a$} (2)
        (4)   edge []    node [above] {$a$} (3)
        (5)   edge []    node [above] {$a$} (4)
        (6)   edge []    node [above] {$a$} (5)
        (7)   edge []    node [above] {$a$} (6)
        (8)   edge []    node [right] {$a$} (7)
        (9)   edge []    node [below]{$a$} (8)
        (10)  edge []    node [below]{$a$} (9)
        (11)  edge []    node [below]{$a$} (10)
        (12)  edge []    node [below]{$a$} (11)
        (13)  edge []    node [below]{$a$} (12)
        (14)  edge []    node [below]{$a$} (13)
        (15)  edge []    node [below]{$a$} (14);

  \path (q0)  edge []    node [above] {0} (0x)
        (q0)  edge []    node [above] {1} (1x)
        (q0)  edge []    node [above] {2} (2x)
        (q0)  edge []    node [above] {3} (3x);

  \path (2x) edge []          node {0} (8)
        (2x) edge []          node [below left] {1} (9)
        (2x) edge []          node [left]  {2} (10)
        (2x) edge []          node [above left] {3} (11)
        (3x) edge []          node {0} (12)
        (3x) edge []          node {1} (13)
        (3x) edge []          node {2} (14)
        (3x) edge []          node [above left] {3} (15);
  
  \path (0x) edge []          node [above,pos=0.4] {0} (0)
        (0x) edge []          node [above right,pos=0.35] {1} (1)
        (0x) edge []          node [right,pos=0.6] {2} (2)
        (0x) edge []          node [below right,pos=0.78] {3} (3)
        (1x) edge []          node [below left,pos=0.78]{0} (4)
        (1x) edge []          node [left,pos=0.6] {1} (5)
        (1x) edge []          node [above left,pos=0.35]{2} (6)
        (1x) edge []          node [above,pos=0.4]{3} (7);
\end{tikzpicture} 
\end{center}
\caption{Strawman example of finite state recognizer accepting strings of $a$ 
characters appended to their lengths expressed as two radix-$4$ digits. 
The initial state $q_0$ accepts the first digit, and memorizes it by means of 
the finite state memory, and moves to one of the four $s_i$ states, which 
in turn read the second digit of the length and transition to an $r_i$ state. 
$r_i$ states are linked by a countdown chain of transitions, where $i$ 
indicates the amount of remaining $a$ characters to be recognized before 
acceptance\label{fig:fsa_example}}
\end{figure}
Providing an effective and efficient grammar definition for 
$\mathbf{L}_{\rm REG-2}$ is similarly difficult. 
Indeed, it is enough to consider a simpler but structurally equivalent 
language $\mathbf{L}_{\rm REG-2-P}$, where only the length constraints on 
DER encoded ASN.$1$ primitive types are enforced, 
to obtain an unmanageable number of productions. 
$\mathbf{L}_{\rm REG-2-P}$ is composed of strings of the form $lv$, 
where $l$ is a byte string containing a unique encoding of a length in 
$[0,2^{126 \cdot 8}]$ and $v$ a string of arbitrary bytes with the given length.
A straightforward grammar representation of it requires at least a grammar 
production for each valid value of the byte string $l$. 
Such a number of productions is well beyond the realm of feasibility 
in realization.
The first observation concerning this issue is that the possible values taken
by the $l$ string in practice will be significantly less than the possible ones,
as the sensible lengths for an X.$509$ certificate are typically far smaller
than $2^{126 \cdot 8}$B.
As a consequence, we deemed reasonable to restrict the set of valid represented
lengths (both in $\mathbf{L}_{\rm REG-2-P}$ and in $\mathbf{L}_{\rm REG-2}$) 
to the $[0,2^{32}-1]$ interval, which entails accepting certificates
of size up to $4$ Gib.
Note that, restricting further the encoded length values to a number which 
allows handwriting of the corresponding linear grammar by a human
designer (i.e., in the hundreds range at most) would critically reduce the 
usefulness of the parser, as it would be able to match only contents smaller
than a couple of hundreds bytes.
We chose the aforementioned conservative upper bound of $2^{32}$ to prevent 
a parser that preallocates memory for the contents of the soon-to-be-recognized 
string in advance from being subject to easy Denial-of-Service attacks, 
via incorrect certificates overstating their lengths.
Since the bound imposed by the practical usefulness on the number of rules of 
the grammar is still insufficient to allow a manageable specification of the 
grammar itself, we resort to analyze the parser automaton of 
$\mathbf{L}_{\rm REG-2-P}$ to derive a more manageable implementation of the 
recognizer, and adapt it to recognize the entire $\mathbf{L}_{\rm REG-2}$.

To illustrate how the problem of implementing the recognizer for 
$\mathbf{L}_{\rm REG-2-P}$ is solved in practice, we first provide a toy 
example, recognizing the language of words $lv$, where the content length $l$ is 
expressed by two, radix-4 digits, and is followed by the content $v$ which is a 
sequence of $a$.

The toy Finite State Automaton (FSA) recognizer operating on the set of bytes 
$\Sigma=\{0,1,2,3,a\}$ is defined as $\mathcal{A}_{\rm REG-2-P}: 
(\Sigma, \mathbf{Q}, \delta, q_0, \mathbf{F})$,
where $\mathbf{Q}$ is the set of states of the automaton, $q_0$ the initial 
state, $\mathbf{F}$ the set of final states, and 
$\delta: \mathbf{Q}$$\times$$\Sigma$$\to$$\mathbf{Q}$ the transition function.
A depiction of the recognizer FSA is provided in Figure~\ref{fig:fsa_example}.
The key idea underlying its functioning is count the $a$ symbols in $v$ 
employing a chain of states, of which only the last is an accepting one 
(depicted in figure as $r_{i}, i\in[0,15]$).
The recognition of the length $l$ is performed while decoding which one is the 
state of the chain of counting states where the computation should start
recognizing the content $v$. The decoding is performed memorizing the number
read digit by digit, yielding a single sequence of transitions starting from the
initial state $q_0$ and ending in a counting state for each one of the latter.

The main issue in implementing an automaton with the structure depicted in 
Figure~\ref{fig:fsa_example} is to provide a way to implement the transition 
function $\delta$ different from the common tabulated form, which would require 
an amount of space proportional to the maximum length of the represented string 
multiplied by the alphabet size; i.e., a number of entries 
$\geq 2^{32}$$\cdot$$256$ for the recognizer of $\mathcal{A}_{\rm REG-2-P}$.

Our approach to cope with this issue is to employ a computational form of the 
said $\delta$, where the states are encoded as unsigned integers, and the 
destination state can be computed in closed form.
In the toy example, the $\delta$ function, which has its description as
$\delta(q_0,t) = s_t, t\in \{0,1,2,3\};\ \delta(s_i,t) = r_{4i+t},$ $t\in 
\{0,1,2,3\};\ \delta(r_i,t) = r_{i-1}, t=a, i \in [1,15]$, can be transformed in 
computational form encoding the $r_i$ states with the unsigned integer $i$, 
the $s_i$ states with $2^4+i$ and $q_0$ as $2^5$, and computing it as
$$
\begin{array}{lll}
\delta(q_0,t) = s_t, t\in \{0,1,2,3\}       &\quad \to \quad \quad & 
\delta(q,t) = 2^4+t,\ t\in \{0,1,2,3\}, q=2^5 \\
\delta(s_i,t) = r_{4i+t}, t\in \{0,1,2,3\}  &\quad \to \quad       & 
\delta(q,t) = (q - 2^4)\cdot 4+t,\ t\in \{0,1,2,3\}, q=2^4+i, i\in [0,3]\\
\delta(r_i,t) = r_{i-1}, i \in [1,15]       &\quad \to \quad       & 
\delta(i,t) = i-1,\ t=a, i \in [1,15]\\
\end{array}
$$
\noindent Following the same line of reasoning employed to obtain a 
computational form for the delta of our toy example automaton, we devised the 
representation of the transition function for the recognizer 
$\mathcal{A}_{\rm REG-2-P}$.
In particular, such a length function deals with the possible variable length
encoding of the length field dedicating a specific transition path for each
one of the lengths:
$$
\delta(q_0,a)=\left\{\begin{array}{l}
a      \ \mathrm{with\,} \text{$0$$\leq\,$$a\,$$\leq$$127$}\\ 
2^{32} \ \mathrm{with\,} \text{$a$$=$}129\\
2^{33} \ \mathrm{with\,} \text{$a$$=$}130\\  
2^{34} \ \mathrm{with\,} \text{$a$$=$}131\\  
2^{35} \ \mathrm{with\,} \text{$a$$=$}132
\end{array}\right.
$$
\noindent This split in the parsing paths also allows us to encode in the 
state the actual value of the length counter, accumulating it in case it is 
longer than a single byte.
All the length-field recognition paths terminate with a transition moving to the 
state encoding the corresponding number of {\em content octets} field bytes to 
be read, upon reading the last byte of the length field, following the same
structure of the toy automaton.
Such a transition function is shown in Figure~\ref{fig:transFunction}.

\begin{figure}[!t]
	\centering
\begin{tabular}{crl}
$\phantom{M}$  & 
$
\begin{array}{c}
\forall\ a\in \term,\\ 
0 \leq i \leq 255\\ 
\end{array}
$

& 

$\delta(q,a)=\left\{
\begin{array}{ll}
a                                                                & 
\mathrm{\ with\ } \text{$q$$=$$2^{32}$,\, $128$$\leq$$a$$\leq$$255$ }\\[3pt]
\text{$2^{33}$$+$$a$$+$$1$}                                      & 
\mathrm{\ with\ } \text{$q$$=$}2^{33}                                \\ 
\text{$a$$+$$(q$$-$$(2^{33}$$+$$1))$$\cdot$$2^8$}                & 
\mathrm{\ with\ } \text{$q$$=$}\delta(2^{33},i)                      \\[3pt]

\text{$2^{34}$$+$$1$$+$$a$}                                      & 
\mathrm{\ with\ } \text{$q$$=$}2^{34}                                \\ 
\text{$2^{34}$$+$$1$$+$$a$$+$$(q$$-$$2^{34}$$-$$1)$$\cdot$$2^8$} & 
\mathrm{\ with\ } \text{$q$$=$}\delta(2^{34},i)                      \\
\text{$a$$+$$(q$$-$$2^{34}$$-$$1)$$\cdot$$2^8$}                  & 
\mathrm{\ with\ } \text{$q$$=$}\delta(\delta(2^{34},i),i)            \\[3pt]

\text{$2^{35}$$+$$1$$+$$a$}                                      & 
\mathrm{\ with\ } \text{$q$$=$}2^{35}                                \\ 
\text{$2^{35}$$+$$1$$+$$a$$+$$(q$$-$$2^{35}$$-$$1)$$\cdot$$2^8$} & 
\mathrm{\ with\ } \text{$q$$=$}\delta(2^{35},i)                      \\
\text{$2^{35}$$+$$1$$+$$a$$+$$(q$$-$$2^{35}$$-$$1)$$\cdot$$2^8$} & 
\mathrm{\ with\ } \text{$q$$=$}\delta(\delta(2^{35},i),i)            \\
\text{$a$$+$$(q$$-$$2^{35}$$-$$1$$)$$\cdot$$2^8$}                & 
\mathrm{\ with\ } \text{$q$$=$}\delta(\delta(\delta(2^{35},i),i),i)  \\
\end{array}\right.$\\
\end{tabular}
\caption{Transition function for the recognizer 
$\mathcal{A}_{\mathrm{REG}-2-\mathrm{P}}$\label{fig:transFunction}}
\end{figure}

Finally, the remaining portion of the transition function takes care of counting
the \emph{content octet} field bytes while they are being read simply 
decrementing the counter encoded in the state as 
$\forall\ a$$\in$$\term,\ \delta(q,a)$$=$$q$$-$$1$ with  
$1$$\leq\,$$q\,$$\leq$$2^{32}$$-$$1$
until the only final state of the automaton $\mathbf{F}$$=$$\{0\}$ is reached.
All the unspecified transitions are leading to the error state.
This $\delta$ can be conveniently implemented in any programming language
employing a single $64$-bit unsigned integer as storage for the current state.

In order to recognize $\mathbf{L}_{\rm REG-2}$ to its full extent, we build on
$\mathcal{A}_{\rm REG-2-P}$ a machine able to take into account ASN.$1$ 
constructed types. 
Since the content of a constructed type is a list of other ASN.$1$ types, 
a nesting structure is introduced, requiring different counters depending 
on the nesting level.
In our case, it is still possible to do so, while still employing a FSA 
recognizer as the maximum nesting depth is bounded, as a consequence of the
capping imposed on the value of the outermost counter.
As a consequence, we represent the state of $\mathcal{A}_{\rm REG-2}$ as a 
vector of states of $\mathcal{A}_{\rm REG-2-P}$, employing one for each nesting 
level, and transitioning from one to the other upon recognizing the beginning
and end of a constructed type.
It is thus possible to implement the recognizer FSA for $\mathbf{L}_{\rm REG}$
taking care of performing the product among the automata 
$\mathcal{A}_{\rm REG-1}$, 
$\mathcal{A}_{\rm REG-2}$ and $\mathcal{A}_{\rm REG-3}$ employing any 
parser generator allowing to do so (e.g., multistate lexers in GNU 
Flex~\cite{flex}, or predicate grammars in 
ANTLR~\cite{DBLP:conf/pldi/ParrF11}).

\subsection{Recognizer Implementation}
Willing to implement the proposed recognizer by means of a tool generating
the actual code from a synthetic-grammar based description of the 
$\mathbf{L}_{\rm X.509-r2}$ language, we choose the formalism of 
\emph{predicated grammars}~\cite{predicated-ll} which can be employed as an 
input to the ANTLR parser generator to obtain an efficient LL(*) 
parser~\cite{DBLP:conf/pldi/ParrF11}.
The LL(*) grammars are obtained as an extension of the common subset of the 
DCF grammars known as LL(k) grammars, i.e., the ones amenable to recursive 
descent parsing without backtracking~\cite{Harrison:1978:IFL:578595}.
Our choice of an LL(*) predicated grammar is justified by the expressiveness of 
the formalism, which allows us to combine the computational
implementations of the transition functions of the length validating automaton
$\mathcal{A}_{\rm REG-2}$, together with the two CS checks and the generated 
portion of the parser via the so-called \emph{semantic predicates}.
We note that, although it is possible in principle to provide an equally 
functional description of the X.$509$ employing a GNU Flex generated parser, 
exploiting the multi-state lexer functionality to combine the different 
transition functions and through a generous use of the semantic actions, such a
description is less manageable as it would be made of a huge regular expression 
defining $\mathbf{L}_{\rm REG}$.
Moreover, LL(*) parsers inherit from LL(k) ones the capability of providing a 
useful syntax error reporting, which is useful to diagnose the nature of 
certificate faults.

Predicated grammars are defined as augmented LL(k) grammars 
adding to the usual quadruple (see Section~\ref{sec:background}) 
$G : (\mathbf{V}_t, \mathbf{V}_n, \mathbf{P}, S, \Pi, \mathbf{M})$ two sets: 
the set of side-effect free \emph{semantic predicates} $\Pi$ and the set of 
\emph{mutators} $\mathbf{M}$.
Semantic predicates are employed as a prefix to the right hand side of a 
production, so that such a production is considered during the 
corresponding recursive descent parsing action only if the predicate evaluated
on the state of the parser is true~\cite{predicated-ll}.
Mutators are the way predicated grammars formalize semantic actions acting on 
the state of the parser, which should be taken upon the matching of the right 
hand side of a rule~\cite{DBLP:conf/pldi/ParrF11}.

The parser generation process of ANTLR produces a modified recursive
descent LL parser, where in lieu of the common sets of finite-length lookaheads
the decisions are taken relying on the fact that the language of the lookaheads
(i.e., the language of the suffixes of the strings) can be split into a 
\emph{regular partition}.
This allows the LL(*) parser to employ FSAs to recognize the lookaheads, 
effectively enhancing its recognition power, although at the possible cost of
examining the entire remainder of the string for each decision to be taken.
Therefore, while this tends to happen quite rarely, it is possible for a parsing 
action to take quadratic time in the size of the input.

In case the input grammar cannot be recognized by an LL(*) automaton, ANTLR 
falls back to an exponential-time backtracking based strategy for the generated 
parser.
ANTLR also supports \emph{semantic actions} to manipulate the results of the 
parsing action, after a successful match of the right hand side of a rule, in 
the same style as GNU Flex~\cite{flex} and GNU Bison~\cite{bison}.

We implemented the recognizer described in the previous sections, employing
semantic actions to perform the CS checks required to match 
$\mathbf{L}_{\rm AId-match}$ and $\mathbf{L}_{\rm SI-match}$, introducing them 
as soon as the portion of the parse tree matched contains the required 
information.
The implementation of the $\delta$ functions for $\mathbf{A}_{\rm REG-1}$ and 
$\mathbf{A}_{\rm REG-2}$ FSAs and the product with the FSA portion of the  
LL(*) matcher automaton are performed exploiting semantic predicates for 
$\mathbf{A}_{\rm REG-2}$ with the aim to forbid the transitions which would be 
valid according to the matcher of $\mathbf{A}_{\rm REG-3}$ alone, and 
performing the uniqueness check for $\mathbf{A}_{\rm REG-1}$ in an appropriate 
semantic action, upon matching of the required portion of the parse tree.
We also employed semantic predicates to obtain an efficient implementation of
the choice of the recognition path in the parser, when we employed the
value of the OIDs to disambiguate whether an \texttt{OCTET STRING} contained
in an \texttt{extnValue} field of the \texttt{Extension} ADT is indeed 
composite or primitive as described in Section~\ref{subsec:undec}.
The same strategy was applied to disambiguate \texttt{BIT STRING} encodings 
in the {\tt subjectPublicKeyInfo} and the {\tt signatureValue} fields.
The recognizer generation process of ANTLR confirmed that no fallback to the
backtracking based parser is made, and that the grammar is indeed LL(*) and non
ambiguous.

LL(*) parsers achieve the same correctness guarantees of classical 
automatically generated parsers: that is, they recognize the same language 
generated by the grammar from which the parser was derived. 
Therefore, we can claim our parser is correct, given that the grammar 
specification is compliant with the standard~\cite{RFC5280}.
We validate the compliance of the grammar we propose by classifying the 
parsing errors reported, exploiting the accurate information obtainable with
an LL(*) recognizer, and by manually inspecting the certificate to verify
that there is actually the syntactic failure reported by our parser. 
We are able to report that no syntactic error detected by other libraries
escapes our parser on the entire corpus of $11$M X.$509$ certificates in use for 
TLS.
%
%
\section{Experimental evaluation}\label{sec:expres}
We conduct an experimental campaign to validate the effectiveness
of our parser against a comprehensive certificate set and to compare it to
the recognition capabilities of $7$ different TLS libraries chosen among the
most widely used in securing HTTPS transactions.
Our systematically generated parser allowed us to discover a significant amount 
of parsing errors, among which security critical issues are present.
We validate such a criticality showing how one of the parsing
flaws allows a common user owning a leaf certificate to act as a CA, effectively
forging valid and trusted certificates for any subject.
Differently from the aforementioned TLS libraries, 
our analysis does not perform any semantic check on the values of the fields 
of the considered X.$509$ certificates, as it is out of the scope of this work.

\compactpar{Experimental Settings and Libraries.} 
In our experimental campaign we generated the LL(*) recognizer employing 
the \texttt{C} backend of ANTLR~\cite{ANTLR3C} version $3.5.1$, as no backends 
for statically compiled languages are available in the most recent ANTLR v$4$.
The choice of employing the \texttt{C} backend in ANTLR was made to allow an
easier future integration of the generated parser within the most common 
programs employing TLS libraries, and to provide a parser implementation with 
good parsing performance.
\begin{figure}[!t]
\begin{center}
\subfloat[Example of two chains for certificates $\mathcal{C}_3$ and its CA 
certificate $\mathcal{C}_2$\label{fig:differential:chains}]{
\begin{tabular}{c}
\includegraphics[width=0.3\linewidth]{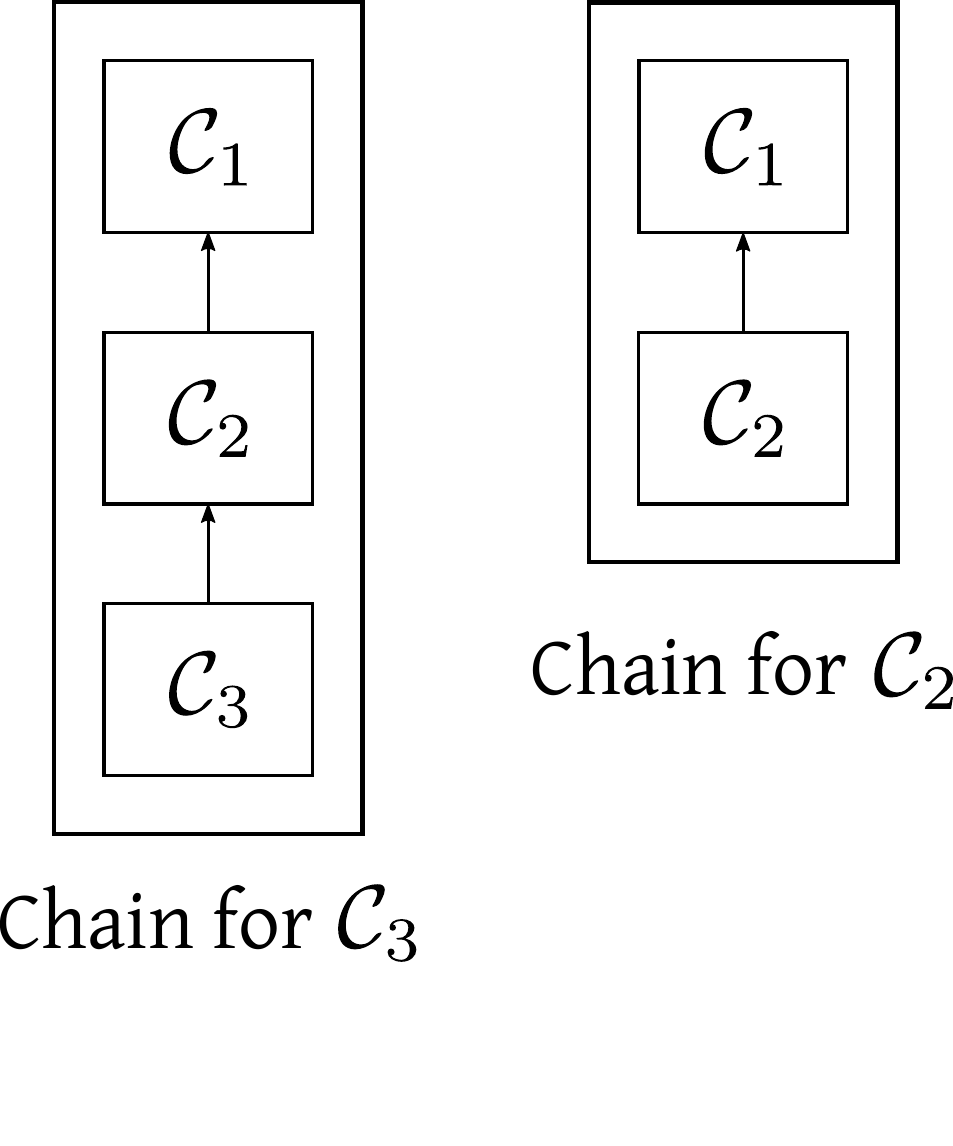}\\[-24pt]
\end{tabular}
}
\hspace{1cm}
\subfloat[Validity outcome for the differential analysis of the certificate 
$\mathcal{C}_3$, given the validity of its chain, and the one of its CA. 
{\bf Error $\mathbf{x}$} and {\bf Error $\mathbf{y}$} are two distinct errors 
among the possible ones\label{fig:differential:result}]{
\begin{tabular}{cc|ccc}
& & \multicolumn{3}{c}{$\mathcal{C}_2$ {\bf chain validity}}             
\\[3pt] 
& & {\bf Valid}  & {\bf Error  $\mathbf{x}$} & {\bf Error  $\mathbf{y}$} \\

\toprule

\parbox[t]{2mm}{\multirow{5}{*}{\rotatebox[origin=c]{90}{
$\mathcal{C}_3$ {\bf chain validity}}}} 
& \multirow{2}{*}{\bf Valid} & \multirow{2}{*}{{\em Valid}}
   & \multicolumn{1}{c}{\multirow{2}{1.5cm}{{\em Does not happen}}}
   & \multicolumn{1}{c}{\multirow{2}{1.5cm}{{\em Does not happen}}}\\[3pt]
   &  &  &  &                                                      \\[3pt]
& \multirow{2}{*}{\bf Error  $\mathbf{x}$}
& \multirow{2}{*}{\em Invalid}     
& \multirow{2}{*}{\color{blue}{\em Valid}} 
& \multirow{2}{*}{\em Invalid}                                     \\[3pt]
&                            &
                             &
                             &                                     \\[3pt]
& \multirow{2}{*}{\bf Error $\mathbf{y}$}
                             & \multirow{2}{*}{\em Invalid}
                             & \multirow{2}{*}{\em Invalid}
                             & \multirow{2}{*}{\color{blue}{\em Valid}}  \\[6pt]
& \multicolumn{4}{c}{\ }\\[12pt]
\end{tabular}
}
\caption{Differential chain analysis approach employed to prevent errors taking
place in an intermediate certificate of a chain from overshadowing others.
Sub-figure (a) reports two samples of certificate chains, the first having as 
leaf $\mathcal{C}_3$, the second one having its CA, $\mathcal{C}_2$ as leaf.
Sub-figure (b) reports a summary of the validation outcome assigned to the leaf 
certificate of a chain, according to the validity reported by a library on both 
the chain terminating with it, and the one terminating with its CA
\label{fig:differential}}
\end{center}
\end{figure}

We compare our parser against the following libraries: OpenSSL 
(v$1$.$0$.$2$g)~\cite{openSSL}, which is employed in both the Apache and NGINX 
web server SSL/TLS modules; 
BoringSSL (commit: ef7dba6ac, $2016$-$01$-$26$)~\cite{boringSSL}, the fork of 
OpenSSL re-engineered by Google and used as default TLS library in Chrome, 
Chromium and Android system binaries; 
GNUTLS (v$3.3.17$)~\cite{gnuTLS}, the Free Software Foundation TLS library 
employed in the GNOME desktop environment facilities and the Common Unix 
Printing System (CUPS); Network Security Services (NSS) 
(v$3.22.2$ Basic ECC)~\cite{nss}, the TLS implementation of the 
Mozilla Foundation employed in Firefox; 
SecureTransport (v$57031.1.35$)~\cite{secureTransport} and CryptoAPI 
(v$10.0.105$ $86.0$)~\cite{cryptoAPI}, the proprietary implementations of TLS 
by Apple and Microsoft respectively, employed for OS updates by both of them; 
and BouncyCastle (v$1.54$), the default security provider in Android apps since 
v$2.1$~\cite{DBLP:conf/ccs/EgeleBFK13}. 
These libraries cover the overwhelming majority of the employed TLS 
providers: while other ones are available, they are either employed in
more restricted application scenarios, or are obtained as a fork of one
of the analyzed libraries, without changing the X.$509$ parsing logic 
(e.g., LibreSSL, Amazon \texttt{s2n}). 
For each one of the chosen libraries we realized a minimal client able to 
process and validate an X.$509$ certificate given its certification path. 

We ran all the validators, including our parser, on a Linux Gentoo $13.0$ 
amd$64$ host based on a six-core Intel Xeon E$5$-$2603$v$3$ endowed with $32$ 
GiB DDR-$4$ DRAM, save for the CryptoAPI and SecureTransport clients, which 
were run on a Windows $10$ machine, equipped with Intel i$7$ $2600,3.4$ GHz 
$64$-bit processor and $8$ GiB DDR-$3$ DRAM, hosting a VMWare virtual 
machine running Mac OSX $10.10$ Yosemite with $5.3$ GiB RAM. 
We stored the results of the validation by all the recognizers in a MySQL 
database, together with {\tt base64} encoded DER certificates to ease the 
processing.

We choose as our dataset the one provided by the Internet-Wide Scan Data 
Repository~\cite{https13} obtained as a result of a scan of the entire IPv$4$ 
address space on TCP port $443$ on the $17$th of December $2016$, which 
encompasses $10,999,727$ X.$509$ digital certificates employed to secure HTTPS 
connections.
A statistical analysis of the certificates shows that their sizes are in 
$[0.31,\,31.5]$ kiB, with an average of $1.41$ kiB and a standard deviation of 
$0.46$ kiB.

All the validation operations were made indicating the date of the certificate 
collection as the current one to the TLS library to avoid unwanted certificate
rejections due to expiration. 
To the end of providing a fair comparison on the syntactic error 
recognition capabilities of the existing libraries on a single certificate 
we need to cope with their choices in terms of user interface.
First of all, the existing libraries provide two different procedures, one to 
parse the certificate and retrieve the information stored, and another one which 
actually performs the validation of such information (e.g., 
checks signatures, checks expiration date) and builds its certification path. 
However, according to the errors reported by the two routines, it is quite
evident that the syntactic error recognition is distributed between both of 
them.
Consequentially, to provide a comparison of the complete syntactic detection 
capabilities of the libraries, we are forced to run also some or all of their 
semantic checks on the certificate, and its certification path, employing the
validation procedure.  
To allow a comparison with our approach, we analyze the errors raised by these 
routines for each library, and we classify them in three categories: 
\begin{itemize}
\item \emph{syntactic}, which are related to a structural error of the 
certificate (e.g., \texttt{Duplicate Extensions Not Allowed} reported by Bouncy 
Castle). 
These are the relevant errors for the comparison with our parser;
\item \emph{validation}, which are related to the content of the certificate or 
to the certification path building (e.g., \texttt{Cert Expired} in GNUTLS or 
\texttt{Unable to Get Issuer Certificate} in OpenSSL). 
These are not relevant errors for the comparison with our parser, since they 
are out of its scope;
\item \emph{generic}, which are errors whose description is too loose to 
understand the issue they refer to, hardening the classification in one of the 
$2$ previous categories (e.g., the \texttt{kSecTrustResultInvalid} error code 
reported by Apple Secure Transport).
We cannot relate these errors to our parser due to this lack of information on 
their source.
\end{itemize}
\noindent We provide our classification of the errors 
reported by all the libraries in Appendix~A, Table A.1.
\begin{table}[!t]
\centering
\caption{Classification of parsing outcomes for the 
$10,999,727$ X.$509$ certificates stored in the 
{\em Internet-Wide Scan Data Repository} (U. Michigan), 
obtained from the differential analysis of $19,024,812$ 
certificate chains\label{tab:stats}}
\begin{tabular}{lcrcr}
\toprule
\multicolumn{1}{c}{\multirow{4}{*}{\bf Library}}  & 
\multirow{2}{*}{{\bf No. of Accepted}}            & 
\multicolumn{3}{c}{{\bf No. of Rejected Certificates}}                \\ 
\cmidrule{3-5}                                    & 
\multirow{2}{*}{\bf Certificates}                 &  {\bf Syntactic} 
                                                  &  {\bf Validation} 
                                                  &  {\bf Generic }   \\
                                                  &
                                                  &  {\bf Rejections} 
                                                  &  {\bf Rejections}
                                                  &  {\bf Rejections} \\ 
\midrule
OpenSSL~\cite{openSSL} ({\em OpenSSL Foundation})               &  $3,656,634$ 
& $70   $     & $7,342,978$ & $45    $ \\
BoringSSL~\cite{boringSSL} ({\em Google Inc.})                  &  $3,656,906$ 
& $70   $     & $7,342,706$ & $45    $ \\
GNUTLS~\cite{gnuTLS} ({\em Free Software Foundation})           &  $4,686,475$ 
& $4,545$     & $6,308,691$ & $16    $ \\
NSS~\cite{nss} ({\em Mozilla Foundation})                       &  $3,457,412$ 
& $33   $     & $7,520,844$ & $21,436$ \\
SecureTransport~\cite{secureTransport} ({\em Apple Inc.})       &  $3,574,491$ 
& $0    $     & $7,389,531$ & $35,705$ \\
CryptoAPI~\cite{cryptoAPI}({\em Microsoft Corporation})         &  $4,197,421$ 
& $591  $     & $6,801,687$ & $28    $ \\
BouncyCastle~\cite{bouncyCastle} ({\em Legion of Bouncy Castle})&  $4,159,585$ 
& $5,795$     & $6,723,630$ & $110,724$\\
{\bf our systematic parser}                                      &  $8,638,063$ 
& $2,361,664$ &     --      &    --    \\
\bottomrule
\end{tabular}
\end{table}
A second point of the validation process of the libraries which we must cope 
with is that they provide a single answer for the validation of an entire 
certification path. 
While this is correct in terms of compliance to the TLS standard behavior, 
where a certificate is invalid if any of the ones composing its certification
path fails validation, this behavior results in an incorrect parent certificate
shadowing with its errors the potential correctness of the certificates for 
which it acts as a CA.
In order to have a clearer insight of the extent of such a behavior we also 
report an ancillary set of results on the certificate correctness derived via a 
differential analysis strategy, of which the gist is reported in 
Figure~\ref{fig:differential}.

The main idea of such a strategy is that if an error pertains to 
a certificate in a certification path, then the error disappears if the 
certificate is no longer in the path, and this fact can be used as a detection
strategy. 
The differential analysis starts from validating all the possible certification
paths, both the ones having an actual leaf certificate as the last one 
($\mathcal{C}_3$, in Figure~\ref{fig:differential}(a)), and all the ones
having its CAs as the last certificate in the path ($\mathcal{C}_2$, 
in Figure~\ref{fig:differential}(a)).
The results are stored, and the leaf certificate is 
deemed valid or invalid according to the policy reported in 
Figure~\ref{fig:differential}(b).
In particular, a certificate is deemed invalid only if the error reported for 
the chain ending with it differs from the one reported for the chain ending with 
its CA.
In case the errors reported for a certificate and its CA during the
chain validation by a library match, the leaf certificate is deemed valid 
(see the two outcomes highlighted in blue in Figure~\ref{fig:differential}(b)), 
as the differential analysis considers the match in errors as an overshadowing
effect of the error in the CA certificate, which in turn causes an early exit
in the validation procedure.
The only precision penalty paid in this process is to deem a leaf certificate 
valid in case it is actually affected by the same error as its CAs, 
which is not expected to be a common case, 
given the significant variety of possible errors.

However, given the tight interleaving of syntactic and semantic checks in 
validation routines performed by the tested libraries, and the need to feed such 
routines with the whole certification path, we claim that no further information 
can be obtained unless a substantial re-engineering of the codebase of each 
library (when available) is performed. 
We remark that the practice of interleaving syntactic (input) validation and 
semantics, while being profitable from a performance standpoint, was identified 
as a pitfall from a security standpoint, since complex validation checks, 
which are more prone to vulnerabilities, may affect the parsing process. 
Indeed, rejecting the certificates upon parsing alone reduces the attack surface 
of the TLS/SSL libraries, as it does not expose internal calls to un-sanitized 
input.
\begin{figure}[!t]
	\centering
	\begin{tikzpicture}
	,\begin{axis}[xbar stacked,
	xmin=0,
	xmax=235,
	ytick=\empty,
	xtick={0,30,...,210,236},
	enlargelimits=0.02,
	width=0.75\textwidth,
	height=2.5cm,
	symbolic y coords = {\ },
	nodes near coords, 
	nodes near coords align=center,
	every node near coord/.append style={font=\scriptsize\boldmath,color=white},
	xlabel={No. of Certificates $\times$ $10^4$}, 
	area style,
	]
	\addplot [fill, color=cyan!95!black] coordinates 
	{(91.2,\ )};\label{plot:first}
	\addplot [fill, color=cyan!75!black] coordinates 
	{(66.3,\ )};\label{plot:second}
	\addplot [fill, color=cyan!50!black] coordinates 
	{(21.3,\ )};\label{plot:third}
	\addplot [fill, color=cyan!30!black] coordinates 
	{(20.5,\ )};\label{plot:fourth}
	\addplot [fill, color=cyan!15!black] coordinates 
	{(36.8,\ )};\label{plot:fifth}
	\end{axis}
	\end{tikzpicture}
	\caption{Syntactic errors recognized by our parser ($2,361,664$).
		The bars shows the number of certificates suffering from: 
		\ref{plot:first}  {\sc missing {\tt keyIdentifier} in not self-issued 
		cert}; 
		\ref{plot:second} {\sc Bad DNS/URI/email}; 
		\ref{plot:third}  {\sc missing {\tt subjectKeyId}}; 
		\ref{plot:fourth} {\sc {\tt keyUsage} violation on PK algorithm}
		\ref{plot:fifth} includes $41$ different kinds of errors
		\label{tab:our_parser_errors}}
	\label{fig:our_parser_errors} 
\end{figure}
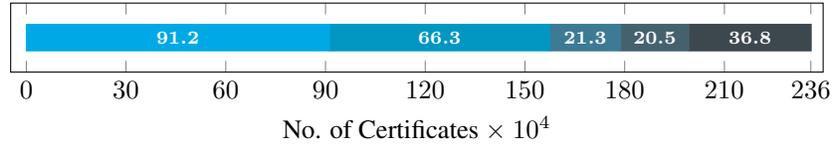

\compactpar{Health State of the X.509 Certificate Set.} 
\noindent For each one of the considered libraries, Table~\ref{tab:stats} 
reports the outcome of the differential analysis of the $19,024,812$ 
certification paths and sub-chains included in the dataset. 
We note that only $31$\%--$43$\% of the entire set is deemed valid by any one of 
the considered libraries, which in turn discard the overwhelming majority of the 
rejected certificates on the basis of validation error in their contents 
(e.g., invalid signature). 
Syntax errors are indeed composing a negligible part of the causes for rejection 
over the whole set of certificates (at most $\approx$ $0.05$\%), confirming the 
concerns on the soundness of the parsing actions of the libraries. 
We note that, due to the interleaving of the syntactic and semantic checking
in common libraries, there is no way to understand which and how many syntax
checks have been run when the semantic validation takes place. As a consequence,
for the certificates deemed semantically invalid, we cannot provide a more 
precise comparison purely between the syntax checks of our parser and the ones 
in the libraries.
A further tiny portion of the certificates (at most $\approx$ $1.6$\%) is 
rejected by the libraries with errors named as \emph{generic}.  

By contrast, our systematically generated parser recognizes $78.5$\% of 
certificates as syntactically X.$509$ abiding and, most interestingly, is able 
to discard $21.5$\% of the whole set of $10,999,727$ certificates via 
an appropriate language analysis  only.
Such a result provides strong evidence supporting a precise parsing action
before the contents of the certificate are employed to perform validation. 

Exploiting the accurate error reporting allowed by LL(*) parsers, we classified
the syntactic issues found. 
Figure~\ref{fig:our_parser_errors} reports errors found by our parser among the 
$2,361,664$ certificates deemed syntactically invalid. 
As shown, only $4$ errors are covering about $84$\% of the rejected 
certificates, that are:
{\sc missing {\tt keyIdentifier} in not self-issued cert}, 
{\sc Bad DNS/URI/email}, {\sc missing {\tt subjectKeyId}}, 
and {\sc {\tt keyUsage} violation on PK algorithm}.
\begin{table*}[!t]
	\caption{Sub-table (a) reports the outcomes for the $8,638,063$ 
	certificates 
		accepted by our parser. The number of certificates having standard 
		features 
		which are unsupported by the tested library is reported in brackets. 
		Sub-table (b) reports the outcomes for the $2,361,664$ 
		certificates rejected by our parser\label{tab:comparison}}
	\captionsetup[subfloat]{position=top}
	\centering
	\subfloat[Libraries outcomes on the set of $8,638,063$ certificates 
	accepted 
	by our 
	parser\label{tab:accepted_certs}]{
		\begin{tabular}{lccrr}
			\toprule
			\multicolumn{1}{c}{\multirow{2}{*}{\bf Library}}
			& \multicolumn{1}{c}{\bf Syntactic}  & \multicolumn{1}{c}{\bf 
			Validation} 
			& \multicolumn{1}{c}{\bf Generic} 
			& \multicolumn{1}{c}{\multirow{1}{*}{\bf Differential}}  \\
			& \multicolumn{1}{c}{\bf Rejection} & \multicolumn{1}{c}{\bf 
			Rejection} 
			& \multicolumn{1}{c}{\bf Rejection} & \multicolumn{1}{c}{\bf 
			Acceptance}  \\
			\midrule
			OpenSSL         & $0$\ {\scriptsize ($8$)}             & 
			$5,613,323$ & $   38 $ & $3,024,694$\\
			BoringSSL       & $0$\ {\scriptsize ($8$)}             & 
			$5,613,122$ & $   38 $ & $3,024,895$\\
			GNUTLS          & $0$\ {\scriptsize $\phantom{(8)}$}   & 
			$5,108,243$ & $    1 $ & $3,529,819$\\
			NSS             & $0$\ {\scriptsize ($2$)}             & 
			$6,081,980$ & $  261 $ & $2,555,820$\\
			SecureTransport & $0$\ {\scriptsize $\phantom{(8)}$}   & 
			$6,003,250$ & $ 2,509$ & $2,632,304$\\
			CryptoAPI       & $0$\ {\scriptsize ($60$)}            & 
			$5,491,470$ & $    0 $ & $3,146,533$\\
			BouncyCastle    & $0$\ {\scriptsize $\phantom{(8)}$}   & 
			$5,508,977$ & $8,341 $ & $3,120,752$\\
			\bottomrule
		\end{tabular}
	}
	\vfill
	
	\subfloat[Libraries outcomes on the set of $2,361,664$ certificates 
	rejected by our parser\label{tab:rejected_certs}]{
		\begin{tabular}{lrcrrrr}
			\toprule
			\multicolumn{1}{c}{\multirow{2}{*}{\bf Library}}
			& \multicolumn{1}{c}{\bf Syntactic} & \multicolumn{1}{c}{\bf 
			Validation} 
			& \multicolumn{1}{c}{\bf Generic} 
			& \multicolumn{1}{c}{\multirow{1}{*}{\bf Differential}} 
			& \multicolumn{1}{c}{\bf Whole Chain} & \multicolumn{1}{c}{\bf 
			Invalid} \\
			& \multicolumn{1}{c}{\bf Rejection}   & \multicolumn{1}{c}{\bf 
			Rejection} 
			& \multicolumn{1}{c}{\bf Rejection}   & \multicolumn{1}{c}{\bf 
			Acceptance}  
			& \multicolumn{1}{c}{\bf Acceptance}  & 
			\multicolumn{1}{c}{\bf Certification Paths}   \\
			\midrule
			OpenSSL          & $  62 $ & $1,729,655$ & $    7 $ & $631,940  $ & 
			$608,838$   & $1,942,174$ \\
			BoringSSL        & $  62 $ & $1,729,584$ & $    7 $ & $632,011  $ & 
			$ 608,919 $ & $1,942,251$ \\
			GNUTLS           & $4,545$ & $1,200,448$ & $   15 $ & $1,156,656$ & 
			$ 711,663 $ & $3,204,136$ \\
			NSS              & $  33 $ & $1,438,864$ & $21,175$ & $901,592  $ & 
			$ 478,681 $ & $2,758,092$ \\
			SecureTransport  & $   0 $ & $1,386,281$ & $33,196$ & $942,187  $ & 
			$ 483,291 $ & $3,398,365$ \\
			CryptoAPI        & $ 531 $ & $1,310,217$ & $   28 $ & $1,050,888$ & 
			$ 611,524 $ & $3,419,740$ \\
			BouncyCastle     & $5,795$ & $1,214,653$ & $102,383$& $1,038,833$ & 
			$ 605,552 $ & $3,412,395$ \\
			\bottomrule
		\end{tabular}
	}
\end{table*}

Among these errors, the first and the third most frequent ones are not 
security critical, since they report missing fields which are required in order 
to help certification path building, thus the worst outcome may be a chain 
building error. 
By contrast, the second and fourth most frequent errors can be a source of 
security issues.
In particular, malformed DNS, URI or email addresses may be exploited for 
malicious purposes by generating a malformed string which is 
actually deemed valid but interpreted differently by two implementations. 
Practical impersonation attacks exploiting differences in the interpretation of 
the same string among different implementations were proven a realistic threat 
in~\cite{DBLP:conf/fc/KaminskyPS10}.

Concerning the fourth most common flaw in the certificates, violations in the 
\texttt{keyUsage} policy, it allows a public key to be employed in a 
cryptographic operation the corresponding primitive is not designed for 
(e.g., a Diffie-Hellman public key employed to check digital signatures), 
with unpredictable and potentially harmful consequences.

We remark our parser identifies $45$ different errors on the dataset, 
showing a high variety of syntactical issues in the certificates 
analyzed, and a highly specialized error recognition capability. 
Among these different errors, one of them refers to the presence of an 
{\tt AlgorithmIdentifier} ADT with an unrecognized algorithm OID.
Recall that we impose this restriction on our grammar both to get rid of 
undecidable portions of the language, as well as relying only on standardized 
algorithm, which have generally undergone an higher level of scrutiny. 
We determined that such a restriction affects only $23,427$ certificates, 
that is, $0.2\%$ of the dataset.
We thus consider such a restriction acceptable in our 
implementation, especially considering that additional {\tt 
AlgorithmIdentifier}s may be added to the parser 
by an implementor who has specific needs -- e.g., maintenance reasons 
arising from the fact that the ITU or one of its authorized subsidiaries 
approved a supplementary (non-standard) algorithm identifier value.

\compactpar{Comparison of Parsing Effectiveness.} 
Table~\ref{tab:comparison}(a) and Table~\ref{tab:comparison}(b) report a 
comparative analysis on the effectiveness of our parser against the validation 
capability of the other libraries.

The results in Table~\ref{tab:comparison}(a) show that no actual syntax errors 
are detected by the tested libraries on the set of $8,638,063$ certificates 
accepted by our parser.
Indeed, a small number of syntax errors are reported by them; however, 
through manual inspection of these few certificates we confirmed that such 
errors are due to the lack of support for some features of the X.$509$ 
standard by the the said libraries. 
In particular, both OpenSSL and CryptoAPI do not support the recognition 
of some fields of \texttt{generalNames} in \texttt{NameConstraints} extension 
and have only partial support for the algorithms included in the Russian 
standard suite GOST.
Moreover, NSS does not match the OID for the MD$5$ hash 
algorithm unless explicitly forced to do so via a flag set during the 
validation process. While discarding MD$5$ signed certificates is a sound
choice from a semantic standpoint, as MD$5$ is known to be cryptographically 
weak, not matching the OID syntactically diverges from the standard 
recommendations.
The absence of syntactic errors reported by other libraries on 
certificates accepted by our parser is really meaningful, since it entails that 
there are no syntactic flaws, identified by other libraries, which are missed 
by our parser. This is a practical outcome further validating the compliance of 
our designed grammar to the X.$509$ standard.
The high number of certificates accepted by our parser and rejected by the 
libraries due to semantic validation errors confirms the sensible intuition 
that checking the semantics of a certificate content 
(e.g., via a signature check) is a crucial point in its validation process.

Table~\ref{tab:comparison}(b) shows how the tested libraries fare on analyzing
the certificates which our parser deems syntactically incorrect.
The first column of the said table shows that the portion of certificates 
recognized by the available libraries as syntactically incorrect is remarkably 
small (i.e., less than $0.25$\% of $2,361,664$ certificates), 
while a significant amount of the considered set of certificates 
(around $56$\%--$73.2$\%) are detected as invalid during the validation phase of 
the libraries (see the element-wise sum of the $2$nd and the $3$rd column in 
Table~\ref{tab:comparison}(b)).
\begin{table}[!t]
\centering
\caption{Syntactic issues on the $468,052$ certificates of the entire dataset, 
which are differentially accepted by all the tested libraries 
\label{tab:issuesDiff}}
\begin{tabular}{cccc}
\begin{tabular}{lr}
\toprule
\multicolumn{1}{c}{\bf Security Critical}        & {\bf Number} \\
\midrule
\texttt{keyCertSign} in leaf certificates        & \multirow{2}{*}{$1,369$} \\
\phantom{keyi}w/o {\tt basicConstraints}         &  \\[3pt]
\texttt{keyUsage} violation                      & \multirow{2}{*}{$87,524$}\\
\phantom{ki}on PK algorithm                      &     \\[3pt]
\texttt{keyCertSign} in leaf certificates        & $97$ \\[3pt]
Wrong string type                                & $1,289$ \\[3pt]
Char Set Violation                               & $14,155$ \\[3pt]
Bad DNS/URI/email format                         & $341,348$ \\[3pt] 
Lexing Error                                     & $15  $  \\[3pt]
Empty Issuer Distinguished Name                  & $368$   \\[3pt]
Duplicated Extension                             & $1$     \\[3pt]
Unexpected \texttt{NULL}                         & \multirow{2}{*}{$1$}   \\
in \texttt{AlgorithmIdP}                         &         \\[3pt]
OID Arc overflow                                 & $5   $  \\[3pt]
Wrong algorithm                                  & $21  $  \\[3pt]
Invalid date                                     & $122 $  \\[3pt]
\bottomrule
\\[0pt]
\end{tabular}
 & & &
\begin{tabular}{lr}
\toprule
\multicolumn{1}{c}{\bf Non Security Critical} & {\bf Number} \\
\midrule
\texttt{keyCertSign} encoding                 & $2$     \\[3pt]
 Empty value field                            & $3$     \\[3pt]
Wrong OID in Distinguished Name               & $10$        \\[3pt]
\texttt{pathLenConstraint} in not             & \multirow{2}{*}{$15$} \\
critical \texttt{basicConstraints}            &          \\[3pt]
Wrong \texttt{extnId}                         & $24$     \\[3pt]
generic error                                 & $56$     \\[3pt]
missing \texttt{subjectKeyId}                 & $61$     \\[3pt]
not critical \texttt{basicConstraints}        & $65$    \\[3pt]
wrong OID                                     & $83$     \\[3pt]
\texttt{pathLenConstraint}                    & \multirow{2}{*}{193}   \\
\phantom{pa}in leaf certificates              &    \\[3pt]
empty \texttt{generalNames}                   & $266$    \\[3pt]
empty string                                  & $401$    \\[3pt]
invalid Distinguished Name                    & $2,575$  \\[3pt]
missing \texttt{keyIdentifier}                &  \multirow{2}{*}{$17,983$}\\  
\phantom{i}in self-issued certificate	      &     \\
\bottomrule
\end{tabular}
\end{tabular}
\end{table}
These results confirm the fact that the certificate validation performed by
libraries tightly blends syntactic validation and semantic checks, instead of 
clearly splitting the two phases.
Nonetheless, even the richness of semantic checking is not sufficient for the 
tested libraries to detect all syntactically invalid certificates.
Indeed, around $26.8$\%--$44$\% of the set of $2,361,664$ syntactically 
incorrect certificates are accepted by the tested libraries 
employing the differential analysis technique described before 
($4$th column in Table~\ref{tab:comparison}(b)).
Such a result can be fruitfully interpreted comparing it with the one of 
the rejections caused on the entire certification path of a given certificate
($5$th column in Table~\ref{tab:comparison}(b)).
Indeed, we recall that differential analysis accepts a certificate 
even if the outcome of the library validation flags it as invalid, if the
reason for the rejection is the same as another certificate on its certification
path.
Consequentially, comparing the results of the whole chain acceptance with the 
differential acceptance ones points to the possibility that the effectiveness in 
rejecting syntactically incorrect certificates of existing libraries is a 
consequence of the rejection of a different, invalid certificate present in the 
same certification path of a syntactically wrong certificate.
Indeed, the only cases where the differential acceptance may underestimate the
recognition capability of the libraries are the ones where in a certification
path two certificates with the same flaw are present.
Nonetheless, even considering the maximum possible selectivity for the existing 
libraries, i.e., deeming a certification path invalid as a whole, between 
$20.2$\% and $30.3$\% of the syntactically invalid certificates are still 
deemed valid by existing libraries.

Finally, to quantify the extent of the missing rejections in the existing 
libraries we computed the set of certificates which are deemed differentially 
accepted by them, but rejected by our parser. Having derived this set of 
certificates we counted how many of the certificates in the entire certificate 
collection are either in the aforementioned set, or are issued by one of the 
subjects of the certificates of the said set (i.e., a certificate in the 
differentially accepted wrong certificates set is their ancestor in a 
certificate path).
The last column of Table~\ref{tab:comparison}(b) reports the amount of such 
certificates for each library, providing the number of certification paths
where at least one certificate is deemed invalid by us, and valid by the 
existing libraries, i.e., the amount of leaf certificates which are deemed valid
when, according to the recursive validation strategy, they shouldn't be.

\subsection{Detected Security Vulnerabilities}
\begin{table}[!t]
\centering
\caption{Syntactic issues on the $468,087$ certificate chains accepted by each 
one of the tested libraries\label{tab:issuesCorrect}}
\begin{tabular}{cccc}
\begin{tabular}{lr}
\toprule
\multicolumn{1}{c}{\bf Security Critical}        & {\bf Number} \\
\midrule
\texttt{keyUsage} violation             & \multirow{2}{*}{$83,033$}\\
\phantom{ki}on PK algorithm                      &     \\[3pt]
\texttt{keyCertSign} in leaf certificates& $1$ \\[3pt]
Wrong string type                       & $81$                 \\[3pt]
Char Set Violation                      & $12,167$              \\[3pt]
Bad DNS/URI/email format                & $366,536$ \\[3pt]
Lexing Error                            & $13  $               \\[3pt]
Invalid Date                            & $1$     \\[3pt]
\bottomrule
 & \\[28pt] 
\end{tabular}
& & &
\begin{tabular}{lr}
\toprule
\multicolumn{1}{c}{\bf Non Security Critical} & {\bf Number} \\
\midrule
Wrong OID in Distinguished Name         & $2$     \\[3pt]
Empty value field                       & $2$     \\[3pt]
empty \texttt{generalNames}             & $5$    \\[3pt]
missing \texttt{subjectKeyId}           & $6$    \\[3pt]
\texttt{pathLenConstraint} in not       & \multirow{2}{*}{$15$}    \\
critical \texttt{basicConstraints}      &          \\[3pt]
\texttt{basicConstraints} not critical  & $38$    \\[3pt]
empty string                            & $199$    \\[3pt]
invalid Distinguished Name              & $1,145$  \\[3pt]
missing \texttt{keyIdentifier}          & \multirow{2}{*}{$4,852$} \\
\phantom{i}in self-issued certificate	      &     \\
\bottomrule	
\end{tabular}\\
\end{tabular}
\end{table}
Table~\ref{tab:issuesDiff} reports the outcomes of our certificate parsing 
on the set of $468,052$ certificates deemed differentially valid by all the 
other libraries (i.e., on the intersection of the sets accounted for in the 
$4^{th}$ column of Table~\ref{tab:comparison}(b)), while the data reported in 
Table~\ref{tab:issuesCorrect} represents the outcomes of our parser on 
certificates belonging to a chain correctly validated by all the 
considered libraries (i.e., on the set of $468,052$ certificates obtained as 
the intersection of the sets accounted for in the penultimate column of 
Table~\ref{tab:comparison}(b)).
Despite the reduction in the variety of errors reported in 
Table~\ref{tab:issuesCorrect} with respect to the ones in 
Table~\ref{tab:issuesDiff}, we note that the errors appearing only in the latter
are not guaranteed to be safely detected by each library, while the ones in 
Table~\ref{tab:issuesCorrect} are definitely evading detection by all the 
libraries. 
\footnote{The outcomes of the syntactic analysis performed by our parser on 
the set of certificates accepted by the considered libraries separately, 
are reported in Appendix~A, Table~\ref{tab:OpenSSLBoringSSL} 
for OpenSSL and BoringSSL, Table~\ref{tab:GNUTLS} for GNUTLS, 
Table~\ref{tab:SecureTransport} for Secure Transport, 
Table~\ref{tab:SBouncyCastle} for Bouncy Castle, 
Table~\ref{tab:cryptoAPI} for CryptoAPI, and Table~\ref{tab:NSS} for NSS}

Both tables categorize the errors splitting them in two groups, 
according to their potential in generating
exploitable security vulnerabilities. In the following, we detail
the possible exploitation for each of these security critical errors, providing
a rationale for such criticality. 
While we deem relevant the detailed investigation of each of the highlighted
security flaws reported in this work, we provide only a single full 
proof-of-concept attack employing them, as the purpose of this work is to 
build a sound X.$509$ parser preventing all of them altogether. 
We spur the detailed investigation in this regard by releasing 
publicly our parser implementation~\cite{Implementation}.

The first and the third entry in Table~\ref{tab:issuesDiff} report issues on 
the \texttt{keyCertSign} bit being
improperly set (either missing the \texttt{BasicConstraints} extension or in a 
leaf certificate) potentially allowing their subject to act as a malicious CA.
We actually prove this attack to be feasible against some of the libraries, 
using certificates which exhibit this kind of error.
Character set violations and strings of a type different from the one expected 
may be exploited in the same way as the already discussed issues about malformed 
DNS, URI or email addresses. 
Therefore, these issues may lead to impersonation attacks arising from different 
interpretations of the same string. 
Lexing errors state an incorrect DER encoding of ASN.$1$ structures which may 
lead to a variety of attacks depending on the type of violation. 
Examining more closely some of them, we discovered that some implementations 
consider null character of byte strings such as \texttt{0x02 0x01 0x03 0x00} 
in the DER encoded INTEGER (tag $2$) of length $1$ with value $3$, 
which suggest such implementations may be employing the outcome of the 
\texttt{strlen} C library function with the value of the actual length field, or 
employ C string based input reading functions. 
More serious issues such as flawed checks on length fields may lead to forgery 
attacks such as the ones reported in~\cite{BERserk-NSSattack,BERserk-Sketch}. 
The errors named as ``Empty Issuer Distinguished Names'' imply that a 
certificate has no issuer, which may lead to security issues depending on how 
the issuer of the certificate is retrieved by the implementation.
The presence of duplicated extensions may lead to impersonation attacks, since
implementations may employ the contents of one duplicate only, at their choice. 
A critical case for such a behavior is the one of a certificate with two 
\texttt{basicConstraints} extensions, one with \texttt{ca} flag set to 
\texttt{true} while the other one set to \texttt{false}. 
In such a case, it depends on an implementation-related choice whether to take 
into account only at the first extension, and thus considering the subject a CA, 
or to consider the second extension valid deeming the subject an end entity. 
Such misinterpretation may lead to powerful attacks where an end entity acts as 
a malicious CA.
\begin{table}[!t]
\centering
\caption{Number of hosts, distinguished via their names, for which 
a certificate deemed incorrect by our parser is accepted by a library, 
considering the result of differential analysis \label{tab:names}}
\begin{tabular}{lrr}
\toprule
Library & \# Certificates & \# Names Affected \\
\midrule
OpenSSL          & $631,940$   & $2,143,635$ \\
BoringSSL        & $632,011$   & $2,143,812$ \\
SecureTransport  & $942,187$   & $3,890,999$ \\
CryptoAPI        & $1,050,888$ & $4,253,895$ \\
GNUTLS           & $1,156,656$ & $4,543,605$ \\
NSS              & $901,592$   & $3,786,986$ \\
BouncyCastle     & $1,038,833$ & $4,196,435$ \\
\bottomrule
\end{tabular}
\end{table}
Unexpected \texttt{NULL} in \texttt{algorithmP} means that there are no 
parameters for an algorithm, even if they are expected, potentially weakening 
the security guarantees provided by some primitives (e.g., missing elliptic 
curve parameters in ECDSA may lead to backtracking to an unsafe default choice).
OID \emph{arc} overflows indicate that a single arc of an OID is indeed 
exceeding the maximum value set by the standard. 
While mishandling of arcs due to a short integer representation is known 
to lead to practical attacks in flawed 
implementations~\cite{DBLP:conf/fc/KaminskyPS10}, we note that an OID arc 
overflow is potentially more dangerous as even standard abiding libraries will 
likely fail to manage it.
Wrong Algorithm errors imply that a non standard algorithm is employed. 
While the employed algorithm may still be a sound one, we note that standardized 
algorithms have usually undergone a higher level of scrutiny.
Finally, we report that some certificates have non-existing dates in their 
validity specification (e.g., the $29$th of February $2022$), potentially 
resulting in an erroneous expiration check. 

\subsection{Analysis of the Certificate Statuses on Distinct Hosts}

Following the quantitative analysis on how many certificates are affected
by potentially security threatening syntactic flaws, we want to analyze 
the effect of such flaws when reflected onto the hosts which are 
employing the said certificates.
To this end, we analyzed the contents of the {\tt Name} ADT of the 
certificates, which usually contains one or more URLs of the host for 
which the certificate is valid.
Moreover, there is a particular standardized extension, called {\tt 
Subject Alternative Names}, which contains a further set of names (we consider
DNS, email addresses, URLs and URIs) which are related to the subject of the 
certificate.
To evaluate the practical impact of flaws in the certificates we report in 
Table~\ref{tab:names} the amount of names bound to flawed certificates
which are not detected by a library, but deemed incorrect by our parser.
As emerging from the data, on average $4$ common names are present for each 
certificate, thus pointing to a four-fold increase in the number of impacted
hosts with respect to the rough estimate which can be provided considering
each certificate as used by a single host.

\subsection{Parsing Vulnerability Exploitation}
Willing to validate the practical exploitation of the security issues emerged
we focus on the syntactic problem of a certificate having no 
\texttt{BasicConstraints} (BC) extension, while having the \texttt{keyCertSign}
bit set in the \texttt{KeyUsage} extension (there are $1,369$ such instances
as reported in Table~\ref{tab:issuesDiff}). Refer to 
Figure~\ref{fig:attack ASN.1} for the detailed structure of these 
two extensions.

The \texttt{BasicConstraints} extension contains a single boolean field 
indicating whether the subject is a CA or not, and an optional constraint on the
maximum length of the certification path.
The \texttt{KeyUsage} extension is a $9$-bit string storing flags which indicate 
the legitimate uses for the public key of the certificate subject.
Semantically, the \texttt{keyCertSign} flag allows 
the public key of the subject to be used to verify certificate signatures, 
but the subject must be a CA. 
Such an information is contained in the \texttt{cA} Boolean field of BC, 
which has a default value of \texttt{FALSE}.
Thus, if the \texttt{BasicConstraints} extension is missing, the subject cannot 
be considered a CA and the certificate must be rejected.
An incorrect behavior in this case was reported in the recent OpenSSL bug 
report~\cite{OpenSSLCABug}: indeed OpenSSL versions prior to $1.0.2$d allowed
such syntactically flawed certificates to be used as CA ones. 
The version of OpenSSL employed in our experimental evaluation
includes the bug fix reported in~\cite{OpenSSLCABug}; however such a fix is not
addressing the root cause of the issue, as a number of syntactically
incorrect certificates are still deemed valid. 
\newbox\listboxb
\begin{lrbox}{\listboxb}
\begin{lstlisting}[language=asn.1]
KeyUsage ::= BIT STRING {
        digitalSignature        (0),
        nonRepudiation          (1),
        keyEncipherment         (2),
        dataEncipherment        (3),
        keyAgreement            (4),
        keyCertSign             (5),
        cRLSign                 (6),
        encipherOnly            (7),
        decipherOnly            (8) }

BasicConstraints ::= SEQUENCE {
        cA                      BOOLEAN DEFAULT FALSE,
        pathLenConstraint       INTEGER (0..MAX) OPTIONAL }
\end{lstlisting}
\end{lrbox}

\begin{figure}[!t]
	\centering
\usebox\listboxb
\caption{\texttt{KeyUsage and \texttt{BasicConstraints} extensions}}
\label{fig:attack ASN.1}
\end{figure} 
We reproduced the issue at hand in a certificate of which we own the private 
key, employing the said private key to sign a forged certificate for 
\texttt{paypal.com}.
The syntactically flawed certificate was signed with a private key of a CA
bootstrapped by us, and such a CA certificate was added to the trusted storages
of the verifying clients, completing the reproduction of the situation found 
in the wild.
We provided the certification path of our forged \texttt{paypal.com} certificate
to OpenSSL both via the programming API, and via command-line client, employing 
the default validation options, succeeding
in getting it accepted as a valid certificate. An interesting remark is that 
such attack is mitigated if the \texttt{x509\_strict} option is set for the
validation algorithm. However, this option is not enabled by default, 
leaving all entities relying on default OpenSSL settings (which are expected to 
be the majority) vulnerable to this serious attack.
We were able to get the said certification path to be accepted also by 
BoringSSL, while the other libraries discard the chain.
Nonetheless, we confirm the practicality of the threat on two of the most used 
TLS implementations, including the one employed by Chrome/Chromium.
As a consequence, any one among the owners of the $1,369$ certificates of the 
dataset with the said vulnerability are likely to be able to exploit it 
successfully, as such items have been signed by a legitimate and trusted 
root CA. 
An inspection of such certificates show that some have been generated with 
the Open Directory framework of Mac OS X Server, providing a pointer 
to a real world TLS implementation generating syntactically incorrect 
certificates. Such implementations can be exploited by an attacker who could 
require a legitimate certificate for its own identity to such implementations. 
The generated certificate would suffer the mentioned syntactic issue and 
could be used as an intermediate certificate to sign other certificates for an 
arbitrary identity. Such certificates can later be used to perform a 
Man-In-The-Middle attack where a TLS session is successfully established 
impersonating the subject of the fake certificate.

\subsection{Performance Analysis}
\begin{figure}[!h]
	\centering
\includegraphics{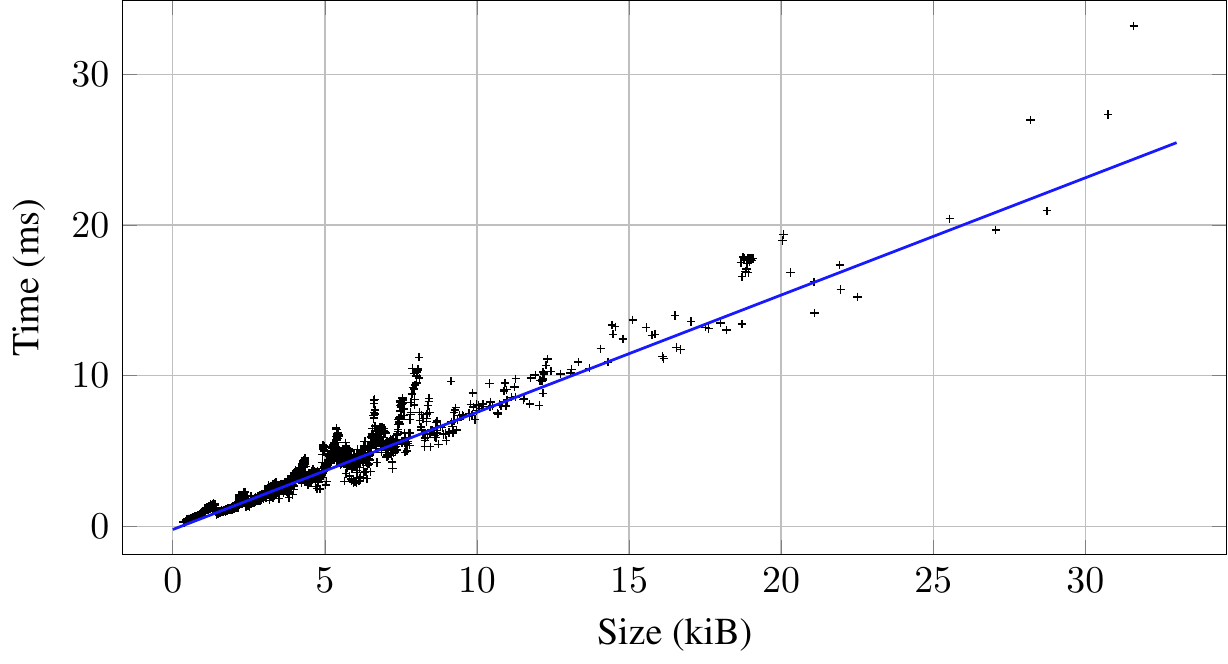}
\caption{Parsing times for all the certificates present in our
dataset. The blue line marks the linear interpolation of the obtained data 
points, with equation $\mathtt{Time}$$=$$0.778$$\cdot$$\mathtt{Size}$$-$$0.203$
\label{fig:performance}}
\end{figure}
Finally, although our work focuses on the effectiveness of the proposed 
certificate parsing approach, we provide some efficiency results of our parser. 
The performance results obtained running our parser on a $3.2$GHz Intel 
i$5$-$6500$ based desktop, running Gentoo Linux $13.0$ (x$86$\_$64$), are 
reported in Figure~\ref{fig:performance}.
The reported timings show how the vast majority of the certificates can be 
recognized in less than $10$ ms, a timing we deem acceptable for practical use, 
and which may be improved employing a more efficient C code generation backend
with respect to the current one in ANTLR$3$.
In addition to the absolute timings reported, it is interesting to note that
the practical parsing complexity appears to be substantially linear, despite
the theoretical quadratic worst case of LL(*) parsers.
Such a fact represents a validation of our claim that the lookahead which the 
LL(*) parser requires is indeed far shorter than the entirety of the remaining
input, indeed resulting in near optimal (i.e., linear) parsing complexity.
%
%
\section{Concluding Remarks}\label{sec:concl}
We presented a systematic approach to the parsing of 
X.$509$ digital certificates, analyzing the standard and providing a description
of the ADTs to be matched in terms of a predicate grammar.
We generated systematically a parser from the given grammar representation, and 
analyzed with it $11$M X.$509$ certificates in use to secure TLS connections.
We report that $21.5$\% are syntactically incorrect and $7$ of the most common
TLS libraries deem authentic $5.7$\%--$10.5$\% of them.
We provided a validation of the practicality of the threat represented by 
mis-parsed certificates, providing a proof of concept of an effective 
impersonation attack stemming from them.
We hope this work will encourage the integration of our systematically generated
parser in widespread TLS libraries, although it is also possible to employ it 
as a stopgap measure in programs dynamically linked to existing libraries,
including it within a wrapper to their API calls. 
Although we provided a single proof-of-concept exploit of a security
vulnerability among the ones we identified, we believe exploring all the 
remaining ones will provide a concrete evaluation of the extent of the current 
issues in X.$509$ certificates.
%
%
\newpage
\bibliographystyle{plainnat}
\bibliography{biblio}
\newpage 
\section*{Appendix A}
\label{sec:appendix}
\setcounter{table}{0}
\renewcommand{\thetable}{A.\arabic{table}}
\begin{table}[!hb]
\caption{Classification of the outcomes of the certificate chain validation 
procedure for each one of the tested TLS libraries as syntactic, validation or 
generic errors\label{table:allErrors}}

{\small
\begin{tabular}{ccc}

\begin{tabular}{cp{0.07\linewidth}p{0.30\linewidth}}
\toprule
\parbox[t]{2mm}{\multirow{6}{*}{\rotatebox[origin=c]{90}{
{\bf SecureTransport$\,$}}}} 
       & Syntactic    & n/a                                           \\[3pt] 
       & Validation   & {\tt kSecTrustResult} deny                     \\
       & Validation   & {\tt kSecTrustResult} recoverable trust failure  \\
       & Validation   & {\tt kSecTrustResult} fatal trust failure        \\[3pt]
       & Generic      & {\tt kSecTrustResult} other error               \\
       & Generic      & {\tt kSecTrustResult} invalid                  \\
\midrule 
\parbox[t]{2mm}{\multirow{7}{*}{\rotatebox[origin=c]{90}{{\bf GNUTLS}}}}
       &  Syntactic   & Insecure Algorithm                            \\[3pt]
       &  Validation  & Certificate signature failure                 \\
       &  Validation  & Certificate signer not CA                     \\
       &  Validation  & Certificate expired                           \\
       &  Validation  & Certificate signer not found                  \\
       &  Validation  & Invalid Certificate                           \\[3pt]
       &  Generic     & Signer constraints failure                    \\

\midrule
\parbox[t]{2mm}{\multirow{15}{*}{\rotatebox[origin=c]{90}{
{\bf OpenSSL, BoringSSL}}}}
       & Syntactic   & Unsupported constraint type                  \\
       & Syntactic   & No explicit policy                           \\[3pt]
       & Validation  & Path length exceeded                         \\
       & Validation  & Permitted violation                          \\
       & Validation  & Certificate not yet valid                    \\
       & Validation  & Unhanded critical extension                  \\
       & Validation  & Unable to get issuer certificate             \\
       & Validation  & Self-signed certificate in chain             \\
       & Validation  & Certificate signature failure                \\
       & Validation  & Certificate has expired                      \\
       & Validation  & Unable to get issuer certificate locally     \\
       & Validation  & Depth zero self-signed certificate           \\[3pt]
       & Generic     & Invalid CA                                   \\
       & Generic     & Invalid policy extension                     \\
       & Generic     & Unable to decode issuer public key           \\
\midrule 
\parbox[t]{2mm}{\multirow{10}{*}{\rotatebox[origin=c]{90}{{\bf NSS}}}} 
       & Syntactic    & Signature algorithm disabled                \\[3pt] 
       & Validation   & Untrusted certificate                       \\
       & Validation   & Bad signature                               \\
       & Validation   & Certificate not in namespace                \\
       & Validation   & Untrusted issuer                            \\
       & Validation   & Expired certificate                         \\
       & Validation   & Unknown issuer                              \\[3pt]
       & Generic      & Invalid arguments                           \\
       & Generic      & Inadequate key usage                        \\
       & Generic      & Inadequate certificate type                 \\
\bottomrule
\end{tabular} 

 & &

\begin{tabular}{lp{0.07\linewidth}p{0.30\linewidth}}
\toprule
\parbox[t]{2mm}{\multirow{14}{*}{\rotatebox[origin=c]{90}{{\bf CryptoAPI}}}}
      & Syntactic    & Not supported name constraint                       \\
      & Syntactic    & NTE bad algorithm identifier                        \\
      & Syntactic    & Invalid date                                      \\[3pt]
      & Validation   & Forbidden name constraint                           \\
      & Validation   & Revoked                                             \\
      & Validation   & Off-line revocation                                 \\
      & Validation   & Weak signature                                      \\
      & Validation   & Signature not valid                                 \\
      & Validation   & Time not valid                                      \\
      & Validation   & Partial chain                                       \\
      & Validation   & Untrusted root                                    \\[3pt]
      & Generic      & Invalid arguments                                   \\
      & Generic      & Inadequate key usage                                \\
      & Generic      & Inadequate certificate type                         \\
\midrule
\parbox[t]{2mm}{\multirow{15}{*}{\rotatebox[origin=c]{90}{
{\bf BouncyCastle$\qquad \qquad \qquad$}}}}
   & Syntactic    & Forbidden extension in v$2$ certificates             \\
   & Syntactic    & Empty {\tt issuerDN}                                     \\
   & Syntactic    & URI must include scheme                                  \\
   & Syntactic    & Empty {\tt subjectDN} is not allowed in v$1$ certificates\\
   & Syntactic    & Invalid URI name                                         \\
   & Syntactic    & Duplicate extensions are not allowed                     \\
   & Syntactic    & Signature algorithm mismatch                             \\
   & Syntactic    & No more data allowed for v$1$ certificate                \\
   & Syntactic    & Incomplete X.$509$ certificate: empty subject field, 
                    and absent {\tt SubjectAlternativeName} extension\\
   & Syntactic    & Incomplete X.$509$ certificate: {\tt SubjectAlternativeName} 
                    extension must be marked as critical, when subject
                    field is empty\\[3pt]
   & Validation   & {\tt targetConstraints} mismatch                         \\
   & Validation   & No issuer found in certification path                    \\
   & Validation   & Unable to find certificate chain                     \\[3pt]
   & Generic      & certificate validation failed                            \\
   & Generic      & Certification path could not be validated                \\
\bottomrule
  & & \\
  & & \\
  & & \\
  & & \\[15pt]
\end{tabular}  
\end{tabular}
}
\end{table}
\begin{table}[!t]
\centering
\caption{Syntactic issues on the $608,838$(resp. $608,919$) 
certificates of the entire dataset differentially accepted by 
OpenSSL(resp. BoringSSL)\label{tab:OpenSSLBoringSSL} }
\begin{tabular}{cc}
\begin{tabular}{lrr}
\toprule
\multicolumn{1}{c}{\bf Security Critical}  & {\bf OpenSSL} & {\bf BoringSSL}\\
\midrule
\texttt{keyUsage} violation  & \multirow{2}{*}{$126,174$} 
                             & \multirow{2}{*}{$126,182$}\\
on PK algorithm              &          &    \\[3pt]
\texttt{keyCertSign} in      & \multirow{2}{*}{$17$} 
                             & \multirow{2}{*}{$10$} \\[3pt]
leaf certificates            &          & \\[3pt]
Wrong string type            & $140$    & $139$ \\[3pt]
Char Set Violation           & $13,913$ & $13,925$ \\[3pt]
Bad DNS/URI/email            & \multirow{2}{*}{$455,539$} 
                             & \multirow{2}{*}{$455,536$} \\[3pt] 
format                       &          & \\[3pt]
Lexing Error                 & $57$     & $57$ \\[3pt]
Wrong algorithm              & $1$      & $11$ \\[3pt]
Extension found but          & \multirow{2}{*}{$1$}
                             & \multirow{2}{*}{$0$}  \\[3pt]
version $ \ne 3$             &          & \\[3pt]
\bottomrule
\\[145pt]
\end{tabular}
 & 
\begin{tabular}{lrr}
\toprule
\multicolumn{1}{c}{\bf Non Security Critical} & {\bf OpenSSL} 
                                              & {\bf BoringSSL}\\
\midrule
Empty value field                             & $2$      & $3$        \\[3pt]
Wrong OID in Distinguished                    & \multirow{2}{*}{$3$}
                                              & \multirow{2}{*}{$2$}  \\[3pt]
Name                                          & & \\[3pt] 
\texttt{pathLenConstraint} in                 & \multirow{2}{*}{$10$} 
                                              & \multirow{2}{*}{$11$} \\
not critical \texttt{basicConstraints}        &          &            \\[3pt]
generic error                                 & $4,909$  & $4,909$    \\[3pt]
missing \texttt{subjectKeyId}                 & $59$  & $80$   \\[3pt]
not critical \texttt{basicConstraints}        & $65$  & $74$   \\[3pt]
wrong OID                                     & $88$  & $89$   \\[3pt]
\texttt{pathLenConstraint}                    & \multirow{2}{*}{$27$} 
                                              & \multirow{2}{*}{$25$}  \\
in leaf certificates                          &         &         \\[3pt]
empty \texttt{generalNames}                   & $14$    & $16$    \\[3pt]
empty string                                  & $278$   & $277$   \\[3pt]
invalid Distinguished Name                    & $1,368$ & $1,370$ \\[3pt]
missing \texttt{keyIdentifier}                &  \multirow{2}{*}{$6,168$} 
                                              & \multirow{2}{*}{$6,196$}\\  
in self-issued certificate	      &       &     \\[3pt]
Bad {\tt BIT STRING} encoding         & $1$   & $1$  \\[3pt]
Redundant Trailing Bytes              & $2$   & $2$  \\[3pt]
not critical \texttt{basicConstraints} & \multirow{2}{*}{$1$} & 
\multirow{2}{*}{$3$}\\			
and missing \texttt{subjectKeyId} &    &   \\[3pt]
Empty sequence in              & \multirow{3}{*}{$1$} & \multirow{3}{*}{$1$} \\
{\tt Authority/Subject}        & & \\
{\tt  Information Access} exts & & \\
\bottomrule
\end{tabular}
\end{tabular}
\end{table}

\begin{table}[!t]
\centering
\caption{Syntactic issues on the $711,663$ certificates of the entire dataset, 
which are differentially accepted by GNUTLS\label{tab:GNUTLS} }
\begin{tabular}{cccc}
\begin{tabular}{lr}
\toprule
\multicolumn{1}{c}{\bf Security Critical}        & {\bf Number} \\
\midrule
\texttt{keyUsage} violation                      & \multirow{2}{*}{$151,727$}\\
on PK algorithm                                  &     \\[3pt]
\texttt{keyCertSign} in leaf certificates        & $20$ \\[3pt]
Wrong string type                                & $397$ \\[3pt]
Char Set Violation                               & $16,637$ \\[3pt]
Bad DNS/URI/email format                         & $514,645$ \\[3pt] 
Lexing Error                                     & $14$       \\[3pt]
Wrong algorithm                                  & $2$      \\[3pt]
Duplicated Extensions                            & $2$      \\[3pt]
\bottomrule
\\[166.5pt]
\end{tabular}
& & &
\begin{tabular}{lr}
\toprule
\multicolumn{1}{c}{\bf Non Security Critical} & {\bf Number} \\
\midrule
Empty value field                            & $2$     \\[3pt]
Wrong OID in Distinguished Name               & $3$        \\[3pt]
\texttt{pathLenConstraint} in not             & \multirow{2}{*}{$12$} \\
critical \texttt{basicConstraints}            &          \\[3pt]
generic error                                 & $4,907$     \\[3pt]
missing \texttt{subjectKeyId}                 & $52$     \\[3pt]
not critical \texttt{basicConstraints}        & $57$    \\[3pt]
wrong OID                                     & $96$     \\[3pt]
\texttt{pathLenConstraint}                    & \multirow{2}{*}{$124$}   \\
in leaf certificates                          &    \\[3pt]
empty \texttt{generalNames}                   & $12$    \\[3pt]
empty string                                  & $487$    \\[3pt]
invalid Distinguished Name                    & $1,638$  \\[3pt]
missing \texttt{keyIdentifier}                &  \multirow{2}{*}{$20,825$}\\  
in self-issued certificate                    &     \\[3pt]
Redundant Trailing Bytes                      & $1$ \\[3pt]
not critical \texttt{basicConstraints}        & \multirow{2}{*}{$1$} \\			
and missing \texttt{subjectKeyId}             &    \\[3pt]
Empty sequence in {\tt Authority/Subject}     & \multirow{2}{*}{$1$} \\
{\tt Information Access} extensions           & \\[3pt]
Wrong \texttt{extnId}                         & $1$   \\[3pt] 	
\bottomrule
\end{tabular}
\end{tabular}
\end{table}

\begin{table}[!t]
\centering
\caption{Syntactic issues on the $483,291$ certificates of the entire dataset, 
which are differentially accepted by Secure Transport
\label{tab:SecureTransport}}
\begin{tabular}{cccc}
\begin{tabular}{lr}
\toprule
\multicolumn{1}{c}{\bf Security Critical}        & {\bf Number} \\
\midrule
\texttt{keyUsage} violation                      & \multirow{2}{*}{$86,276$}\\
on PK algorithm                                  &     \\[3pt]
\texttt{keyCertSign} in leaf certificates        & $4$ \\[3pt]
Wrong string type                                & $84$ \\[3pt]
Char Set Violation                               & $12,558$ \\[3pt]
Bad DNS/URI/email format                         & $377,788	$ \\[3pt] 
Lexing Error                                     & $14$       \\[3pt]
Wrong algorithm                                  & $1$      \\[3pt]
\bottomrule
\\[77pt]
\end{tabular}
& & &
\begin{tabular}{lr}
\toprule
\multicolumn{1}{c}{\bf Non Security Critical} & {\bf Number} \\
\midrule
Empty value field                             & $2$     \\[3pt]
Wrong OID in Distinguished Name               & $2$        \\[3pt]
\texttt{pathLenConstraint} in not             & \multirow{2}{*}{$8$} \\
critical \texttt{basicConstraints}            &          \\[3pt]
missing \texttt{subjectKeyId}                 & $31$     \\[3pt]
not critical \texttt{basicConstraints}        & $85$    \\[3pt]
empty \texttt{generalNames}                   & $5$    \\[3pt]
empty string                                  & $207$    \\[3pt]
invalid Distinguished Name                    & $1,163$  \\[3pt]
missing \texttt{keyIdentifier}                &  \multirow{2}{*}{$5,060$}\\  
in self-issued certificate                    &     \\[3pt]
Redundant Trailing Bytes                      & $2$ \\[3pt]
Empty sequence in {\tt Authority/Subject}     & \multirow{2}{*}{$1$} \\
{\tt Information Access} extensions           &   \\[3pt]
\bottomrule
\end{tabular}
\end{tabular}
\end{table}

\begin{table}[!t]
\centering
\caption{Syntactic issues on the $608,038$ certificates of the entire dataset, 
which are differentially accepted by Bouncy Castle\label{tab:SBouncyCastle} }
\begin{tabular}{cccc}
\begin{tabular}{lr}
\toprule
\multicolumn{1}{c}{\bf Security Critical}        & {\bf Number} \\
\midrule
\texttt{keyUsage} violation                      & \multirow{2}{*}{$126,232$}\\
on PK algorithm                                  &     \\[3pt]
\texttt{keyCertSign} in leaf certificates        & $3$ \\[3pt]
Wrong string type                                & $137$ \\[3pt]
Char Set Violation                               & $13,710$ \\[3pt]
Bad DNS/URI/email format                         & $452,630	$ \\[3pt] 
Lexing Error                                     & $13$       \\[3pt]
Wrong algorithm                                  & $1$      \\[3pt]
\bottomrule
\\[138.5pt]
\end{tabular}
& & &
\begin{tabular}{lr}
\toprule
\multicolumn{1}{c}{\bf Non Security Critical} & {\bf Number} \\
\midrule
Empty value field                             & $2$     \\[3pt]
Wrong OID in Distinguished Name               & $2$        \\[3pt]
\texttt{pathLenConstraint} in not             & \multirow{2}{*}{$8$} \\
critical \texttt{basicConstraints}            &          \\[3pt]
generic error                                 & $4,907$     \\[3pt]
missing \texttt{subjectKeyId}                 & $25$     \\[3pt]
not critical \texttt{basicConstraints}        & $57$    \\[3pt]
wrong OID                                     & $88$     \\[3pt]
\texttt{pathLenConstraint}                    & \multirow{2}{*}{$26$}   \\
in leaf certificates                          &    \\[3pt]
empty \texttt{generalNames}                   & $6$    \\[3pt]
empty string                                  & $274$    \\[3pt]
invalid Distinguished Name                    & $1,355$  \\[3pt]
missing \texttt{keyIdentifier}                &  \multirow{2}{*}{$6,073$}\\  
in self-issued certificate	              &     \\[3pt]
Redundant Trailing Bytes                      & $2$                  \\[3pt]
Empty sequence in {\tt Authority/Subject}     & \multirow{2}{*}{$1$} \\
{\tt Information Access} extensions           &  \\[3pt]
\bottomrule
\end{tabular}
\end{tabular}
\end{table}

\begin{table}[!t]
\centering
\caption{Syntactic issues on the $611,524$ certificates of the entire dataset, 
which are differentially accepted by CryptoAPI\label{tab:cryptoAPI} }
\begin{tabular}{ccc}
\begin{tabular}{lr}
\toprule
\multicolumn{1}{c}{\bf Security Critical}        & {\bf Number} \\
\midrule
\texttt{keyUsage} violation                      & \multirow{2}{*}{$121,686$}\\
on PK algorithm                                  &     \\[3pt]
\texttt{keyCertSign} in leaf certificates        & $10$ \\[3pt]
Wrong string type                                & $194$ \\[3pt]
Char Set Violation                               & $14,105$ \\[3pt]
Bad DNS/URI/email format                         & $462,238$ \\[3pt] 
Lexing Error                                     & $14$       \\[3pt]
Wrong algorithm                                  & $1$      \\[3pt]
Wrong {\tt keyCertSign} ASN.$1$ enc.             & $1$      \\[3pt]
\bottomrule
\\[122pt]
\end{tabular}
& & 
\begin{tabular}{lr}
\toprule
\multicolumn{1}{c}{\bf Non Security Critical} & {\bf Number} \\
\midrule
Empty value field                             & $3$     \\[3pt]
Wrong OID in Distinguished Name               & $2$        \\[3pt]
\texttt{pathLenConstraint} in not             & \multirow{2}{*}{$12$} \\
critical \texttt{basicConstraints}            &          \\[3pt]
generic error                                 & $4,907$     \\[3pt]
missing \texttt{subjectKeyId}                 & $23$     \\[3pt]
not critical \texttt{basicConstraints}        & $94$    \\[3pt]
wrong OID                                     & $105$     \\[3pt]
\texttt{pathLenConstraint}                    & \multirow{2}{*}{$29$}   \\
in leaf certificates                          &    \\[3pt]
empty \texttt{generalNames}                   & $7$    \\[3pt]
empty string                                  & $523$    \\[3pt]
invalid Distinguished Name                    & $1,421$  \\[3pt]
missing \texttt{keyIdentifier}                &  \multirow{2}{*}{$6,146$}\\  
in self-issued certificate	              &     \\[3pt]
Redundant Trailing Bytes                      & $2$       \\[3pt]
Empty sequence in {\tt Authority/Subject}     & \multirow{2}{*}{$1$} \\
{\tt Information Access} extensions &           \\[3pt]
\bottomrule
\end{tabular}
\end{tabular}
\end{table}

\begin{table}[!t]
\centering
\caption{Syntactic issues on the $478,681$ certificates of the entire dataset, 
which are differentially accepted by NSS\label{tab:NSS} }
\begin{tabular}{cccc}
\begin{tabular}{lr}
\toprule
\multicolumn{1}{c}{\bf Security Critical}        & {\bf Number} \\
\midrule
\texttt{keyUsage} violation                      & \multirow{2}{*}{$87,187$}\\
on PK algorithm                                  &     \\[3pt]
\texttt{keyCertSign} in leaf certificates        & $2$ \\[3pt]
Wrong string type                                & $98$ \\[3pt]
Char Set Violation                               & $12,832$ \\[3pt]
Bad DNS/URI/email format                         & $372,132$ \\[3pt] 
Lexing Error                                     & $14$       \\[3pt]
Wrong algorithm                                  & $1$      \\[3pt]
\bottomrule
\\[92.5pt]
\end{tabular}
& & &
\begin{tabular}{lr}
\toprule
\multicolumn{1}{c}{\bf Non Security Critical} & {\bf Number} \\
\midrule
Empty value field                             & $2$     \\[3pt]
Wrong OID in Distinguished Name               & $2$        \\[3pt]
\texttt{pathLenConstraint} in not             & \multirow{2}{*}{$7$} \\
critical \texttt{basicConstraints}            &          \\[3pt]
missing \texttt{subjectKeyId}                 & $20$     \\[3pt]
not critical \texttt{basicConstraints}        & $54$    \\[3pt]
wrong OID                                     & $62$     \\[3pt]
\texttt{pathLenConstraint}                    & \multirow{2}{*}{$1$}   \\
in leaf certificates                          &              \\[3pt]
empty \texttt{generalNames}                   & $6$      \\[3pt]
empty string                                  & $202$    \\[3pt]
invalid Distinguished Name                    & $1,177$  \\[3pt]
missing \texttt{keyIdentifier}                &  \multirow{2}{*}{$4,881$}\\  
in self-issued certificate	              &     \\[3pt]
Redundant Trailing Bytes                      & $1$ \\[3pt]
\bottomrule
\end{tabular}
\end{tabular}
\end{table}
\end{document}